\documentclass[aps, pra, reprint, superscriptaddress, onecolumn, notitlepage, tightenlines, 11pt]{revtex4-2}

\pdfoutput=1

\usepackage{header}

\graphicspath{{graphics/}}

\begin{document}

\bibliographystyle{apsrev4-2}

\title{Algorithmic Cluster Expansions for Quantum Problems}

\author{Ryan L. Mann}
\email{mail@ryanmann.org}
\homepage{http://www.ryanmann.org}
\affiliation{Centre for Quantum Computation and Communication Technology, Centre for Quantum Software and Information, School of Computer Science, Faculty of Engineering \& Information Technology, University of Technology Sydney, NSW 2007, Australia}
\affiliation{School of Mathematics, University of Bristol, Bristol, BS8 1UG, United Kingdom}

\author{Romy M. Minko}
\affiliation{School of Mathematics, University of Bristol, Bristol, BS8 1UG, United Kingdom}

\begin{abstract}
    We establish a general framework for developing approximation algorithms for a class of counting problems. Our framework is based on the cluster expansion of abstract polymer models formalism of Koteck\'y and Preiss. We apply our framework to obtain efficient algorithms for (1) approximating probability amplitudes of a class of quantum circuits close to the identity, (2) approximating expectation values of a class of quantum circuits with operators close to the identity, (3) approximating partition functions of a class of quantum spin systems at high temperature, and (4) approximating thermal expectation values of a class of quantum spin systems at high temperature with positive-semidefinite operators. Further, we obtain hardness of approximation results for approximating probability amplitudes of quantum circuits and partition functions of quantum spin systems. This establishes a computational complexity transition for these problems and shows that our algorithmic conditions are optimal under complexity-theoretic assumptions. Finally, we show that our algorithmic condition is almost optimal for expectation values and optimal for thermal expectation values in the sense of zero freeness.
\end{abstract}

\maketitle

{
\hypersetup{linkcolor=black}
\tableofcontents
}

\section{Introduction}
\label{section:Introduction}

The classification of the computational complexity of quantum problems is important for understanding the capabilities and limitations of quantum computing. These problems include the computation of probability amplitudes, expectation values, partition functions, and thermal expectation values. In this paper we consider the classification of such problems in the sense of approximate counting. We establish a general framework for developing approximation algorithms and hardness of approximation results for a class of counting problems. By applying this framework, we are able to obtain efficient approximation algorithms and hardness of approximation results for several quantum problems under certain algorithmic conditions.

Our results concerning the approximation of quantum problems may be summarised as follows. We obtain efficient algorithms for (1) approximating probability amplitudes of a class of quantum circuits close to the identity, (2) approximating expectation values of a class of quantum circuits with operators close to the identity, (3) approximating partition functions of a class of quantum spin systems at high temperature, and (4) approximating thermal expectation values of a class of quantum spin systems at high temperature with positive-semidefinite operators. Our approach offers a simpler and sharper analysis compared to existing algorithms. Our algorithmic results are summarised in Table~\ref{table:approximationquantumproblemssummary}.

\begin{table}[ht]
    \centering
    \setcellgapes{3pt}
    \makegapedcells
    \begin{tabularx}{\textwidth}{cccc}
        \hline
        \makecell{\textbf{Quantum} \\ \textbf{Problem}} & \makecell{\textbf{Conditioned} \\ \textbf{Object}} & \makecell{\textbf{Algorithmic Condition} \\ \textbf{(This Work)}} & \makecell{\textbf{Algorithmic Condition} \\ \textbf{(Previous)}} \\
        \hline
        \makecell{Probability \\ Amplitudes} & \makecell{Unitary Operators \\ $\{U_\varepsilon\}_{\varepsilon \in E}$} & \makecell{$\norm{U_\varepsilon-\mathbb{I}}\leq\frac{1}{e^3\Delta\binom{r}{2}}$} & \makecell{-} \\
        \makecell{Expectation \\ Values} & \makecell{Self-Adjoint Operators \\ $\{O_v\}_{v \in V}$} & \makecell{$\norm{O_v-\mathbb{I}}\leq\frac{1}{e^3k^{3d}}$} & \makecell{$\norm{O_v-\mathbb{I}}<\frac{1}{60k^{5d}}$~\cite{bravyi2021classical}} \\
        \makecell{Partition \\ Functions} & \makecell{Inverse Temperature \\ $\beta$} & \makecell{$\abs{\beta}\leq\frac{1}{e^4\Delta\binom{r}{2}}$} & \makecell{$\abs{\beta}\leq\frac{1}{16e^4\Delta\binom{r}{2}}$~\cite{kuwahara2020clustering}} \\
        \makecell{Thermal Expectation \\ Values} & \makecell{Inverse Temperature \\ $\beta$} & \makecell{$\abs{\beta}\leq\frac{1}{e^4\Delta\binom{r}{2}}$} & \makecell{$\abs{\beta}\leq\frac{1}{2e^2(\Delta-1)r(\Delta r-r+1)}$~\cite{wild2023classical}} \\
        \hline
    \end{tabularx}
    \caption{Summary of algorithmic results for quantum problems. In the context of probability amplitudes, partition functions, and thermal expectation values, $\Delta$ denotes the maximum degree and $r$ denotes the rank of the interaction multihypergraph. In the context of expectation values, $d$ denotes the depth of the quantum circuit, and $k$ denotes the maximum number of qudits each gate can act on.}
    \label{table:approximationquantumproblemssummary}
\end{table}

Our algorithmic framework is based on the cluster expansion of abstract polymer models formalism of Koteck\'y and Preiss~\cite{kotecky1986cluster}. We consider polymers that are connected subgraphs of bounded-degree bounded-rank multihypergraphs with compatibility relations defined by vertex disjointness. The key insight underlying our framework is that when the polymer weights decay sufficiently fast, computing the truncated cluster expansion to sufficiently high order allows us to obtain a multiplicative approximation to the abstract polymer model partition function. Our framework can be viewed as a straightforward generalisation of the framework of Helmuth, Perkins, and Regts~\cite{helmuth2020algorithmic}, and Borgs et al.~\cite{borgs2020efficient} from the case of bounded-degree graphs to bounded-degree bounded-rank multihypergraphs. This approach is closely related to that of Patel and Regts~\cite{patel2017deterministic} using Barvinok's method~\cite{barvinok2016combinatorics}; see Ref.~\cite{patel2022approximate} for a survey of this method.

Our hardness of approximation framework is based on reductions from the Ising model partition function. We apply this framework to obtain hardness of approximation results for approximating probability amplitudes of quantum circuits and partition functions of quantum spin systems. This establishes a computational complexity transition for these problems and shows that our algorithmic conditions are optimal under complexity-theoretic assumptions. Further, we show that our algorithmic condition is almost optimal for expectation values and optimal for thermal expectation values in the sense of zero freeness.

This paper is structured as follows. In Section~\ref{section:Preliminaries}, we introduce the necessary preliminaries. Then, in Section~\ref{section:GeneralFramework}, we establish our algorithmic and hardness of approximation framework. In Section~\ref{section:Applications}, we apply our framework to several quantum problems. Finally, we conclude in Section~\ref{section:ConclusionAndOutlook} with some remarks and open problems.

\section{Preliminaries}
\label{section:Preliminaries}

\subsection{Graph Theory}
\label{section:GraphTheory}

A \emph{multigraph} is a graph in which multiple edges between vertices are permitted. A \emph{hypergraph} is a graph in which edges between any number of vertices are permitted. A \emph{multihypergraph} is a graph in which multiple edges between vertices and edges between any number of vertices are permitted. We shall assume that the edges in a multihypergraph are uniquely labelled, that is, all edges are considered distinct. Let $G=(V, E)$ be a multihypergraph. We denote the \emph{order} of $G$ by $\abs{G}=\abs{V(G)}$ and the \emph{size} of $G$ by $\norm{G}=\abs{E(G)}$. The \emph{maximum degree} $\Delta(G)$ of $G$ is the maximum degree over all vertices of $G$ and the \emph{rank} $r(G)$ of $G$ is the maximum cardinality of an edge of $G$. The \emph{distance} $d(u, v)$ between two vertices $u$ and $v$ in $G$ is defined as the size of the shortest path connecting them. A multihypergraph is called $\Delta$-\emph{regular} if all the vertices have degree $\Delta$ and called $r$-\emph{uniform} if all the edges have cardinality $r$.

\subsection{Abstract Polymer Models}
\label{section:AbstractPolymerModels}

An \emph{abstract polymer model} is a triple $(\mathcal{C}, w, \sim)$, where $\mathcal{C}$ is a countable set of objects called \emph{polymers}, \mbox{$w:\mathcal{C}\to\mathbb{C}$} is a function that assigns to each polymer $\gamma\in\mathcal{C}$ a complex number $w_\gamma$ called the \emph{weight} of the polymer, and $\sim$ is a \emph{symmetric compatibility relation} such that each polymer is incompatible with itself. A set of polymers is called \emph{admissible} if the polymers in the set are all pairwise compatible. Note that the empty set is admissible. Let $\mathcal{G}$ denote the collection of all admissible sets of polymers from $\mathcal{C}$. The abstract polymer partition function is defined by
\begin{equation}
    Z(\mathcal{C},w) \coloneqq \sum_{\Gamma\in\mathcal{G}}\prod_{\gamma\in\Gamma}w_\gamma. \notag
\end{equation}
The archetypal example of an abstract polymer model is the independence polynomial. Let $G=(V, E)$ be a graph and let $\mathcal{I}$ denote the collection of all independent sets of $G$. Recall that an independent set of $G$ is a subset of vertices with no edges between them. The independence polynomial $I(G;x)$ of $G$ is a polynomial in $x$, defined by
\begin{equation}
    I(G;x) \coloneqq \sum_{I\in\mathcal{I}}x^\abs{I}. \notag
\end{equation}
This corresponds to an abstract polymer model $(\mathcal{C}, w, \sim)$ as follows. The polymers $\mathcal{C}$ are the vertices $V$ of $G$, the weight function $w$ is given by $w_\gamma=x$ for all $\gamma\in\mathcal{C}$, and two polymers are compatible if and only if there is no edge between them in $G$. An admissible set of polymers is then an independent set of $G$, and it follows that the partition function of this model $Z(\mathcal{C},w)$ is precisely the independence polynomial $I(G;x)$ of $G$. In the context of statistical physics, the independence polynomial $I(G;x)$ is the partition function of the hard-core model with activity $x$. In this model, each independent set corresponds to a possible configuration of particles that are subject to hard-core repulsion. The abstract polymer model can be viewed as a generalisation of the independence polynomial. In particular, it attempts to capture the independence properties of a problem.

A useful tool for representing a problem as an abstract model is the principle of inclusion-exclusion. The principle is formalised by the following well-known lemma (see for example \cite[Theorem 12.1]{graham1995handbook}); we provide a proof for completeness.
\begin{lemma}[Principle of inclusion-exclusion]
    \label{lemma:InclusionExclusion}
    Let $f$ be a function defined on the subsets of finite set $E$, then
    \begin{equation}
        f(E) = \sum_{S \subseteq E}(-1)^\abs{S}\sum_{T \subseteq S}(-1)^\abs{T}f(T). \notag
    \end{equation}
\end{lemma}
\begin{proof}
    By interchanging the summations, we have
    \begin{align}
        \sum_{S \subseteq E}(-1)^\abs{S}\sum_{T \subseteq S}(-1)^\abs{T}f(T) &= \sum_{T \subseteq E}(-1)^\abs{T}f(T)\sum_{T \subseteq S \subseteq E}(-1)^\abs{S} \notag \\
        &= \sum_{T \subseteq E}f(T)\sum_{S \subseteq E{\setminus}T}(-1)^\abs{S} \notag \\
        &= \sum_{T \subseteq E}f(T)\sum_{m=0}^\abs{E{\setminus}T}\binom{\abs{E{\setminus}T}}{m}(-1)^m. \notag
    \end{align}
    Now, by applying the binomial theorem, we obtain
    \begin{equation}
        \sum_{S \subseteq E}(-1)^\abs{S}\sum_{T \subseteq S}(-1)^\abs{T}f(T) = f(E), \notag
    \end{equation}
    completing the proof.
\end{proof}

As we shall see in Section~\ref{section:Applications}, several quantum problems admit an abstract polymer model representation.

\subsection{Abstract Cluster Expansion}
\label{section:AbstractClusterExpansion}

We now define the \emph{abstract cluster expansion}~\cite{kotecky1986cluster, friedli2017statistical}. Let $\Gamma$ be a non-empty ordered tuple of polymers. The \emph{incompatibility graph} $H_\Gamma$ of $\Gamma$ is the graph with vertex set $\Gamma$ and edges between any two polymers if and only if they are incompatible. $\Gamma$ is called a \emph{cluster} if its incompatibility graph $H_\Gamma$ is connected. A polymer and cluster are compatible if the polymer is compatible with every polymer in the cluster. Let $\mathcal{G}_C$ denote the set of all clusters of polymers from $\mathcal{C}$. The abstract cluster expansion is a formal power series for $\log{Z(\mathcal{C},w)}$ in the variables $w_\gamma$, defined by
\begin{equation}
    \log(Z(\mathcal{C},w)) \coloneqq \sum_{\Gamma\in\mathcal{G}_C}\varphi(H_\Gamma)\prod_{\gamma\in\Gamma}w_\gamma, \notag
\end{equation}
where $\varphi(H)$ denotes the \emph{Ursell function} of a graph $H$:
\begin{equation}
    \varphi(H) \coloneqq \frac{1}{\abs{H}!}\sum_{\substack{S \subseteq E(H) \\ \text{spanning} \\ \text{connected}}}(-1)^{\abs{S}}. \notag
\end{equation}
The sum is over all spanning connected edge sets of $H$. For a derivation of the cluster expansion, we refer the reader to Ref.~\cite{friedli2017statistical}.

An important theorem due to Koteck\'y and Preiss~\cite{kotecky1986cluster} provides a sufficient criterion for the absolute convergence of the cluster expansion. An improved convergence criterion is given in Ref.~\cite{fernandez2007cluster}.
\begin{theorem}[Koteck\'y and Preiss~\cite{kotecky1986cluster}]
    \label{theorem:KoteckyPreissConvergenceCriterion}
    Let $(\mathcal{C}, w, \sim)$ be an abstract polymer model and let $a:\mathcal{C}\to\mathbb{R}^+$ and $d:\mathcal{C}\to\mathbb{R}^+$ be functions such that
    \begin{equation}
    \sum_{\gamma^*\nsim\gamma}\abs{w_{\gamma^*}}e^{a(\gamma^*)+d(\gamma^*)} \leq a(\gamma), \notag \notag
    \end{equation}
    for all polymers $\gamma\in\mathcal{C}$. Then the cluster expansion for $\log(Z(\mathcal{C},w))$ converges absolutely, $Z(\mathcal{C},w)\neq0$, and
    \begin{equation}
        \sum_{\substack{\Gamma\in\mathcal{G}_C \\ \Gamma\nsim\gamma }}\abs{\varphi(H_\Gamma)\prod_{\gamma^*\in\Gamma}w_{\gamma^*}}e^{\sum_{\gamma^*\in\Gamma}d(\gamma^*)} \leq a(\gamma), \notag
    \end{equation}
    for all polymers $\gamma\in\mathcal{C}$.
\end{theorem}
We shall use Theorem~\ref{theorem:KoteckyPreissConvergenceCriterion} to establish the absolute convergence of the cluster expansion. 

In the case of the independence polynomial, the radius of convergence is given by Shearer's bound for the Lov\'asz Local Lemma~\cite{shearer1985problem}; this was elucidated by Scott and Sokal~\cite{scott2005repulsive}. For results on the hypergraph independence polynomial see Refs.~\cite{galvin2022zeroes, bencs2023optimal}. Note that the Koteck\'y-Preiss convergence criterion can be viewed as a type of local lemma.

Let $\norm{\;\cdot\;}:\mathcal{C}\to\mathbb{Z}^+$ be a function that assigns to each polymer $\gamma\in\mathcal{C}$ a positive integer $\norm{\gamma}$ called the \emph{size} of the polymer. A useful quantity for algorithmic purposes is the truncated
cluster expansion $T_m(Z(\mathcal{C},w))$ for $\log(Z(\mathcal{C},w))$:
\begin{equation}
    T_m(Z(\mathcal{C},w)) \coloneqq \sum_{\substack{\Gamma\in\mathcal{G}_C \\ \norm{\Gamma} \geq m}}\varphi(H_\Gamma)\prod_{\gamma\in\Gamma}w_\gamma, \notag
\end{equation}
where $\norm{\Gamma}=\sum_{\gamma\in\Gamma}\norm{\gamma}$.

It is often convenient to consider clusters as multisets of polymers. Define a multiset cluster to be a multiset $(\Gamma, m_\Gamma)$ of polymers $\Gamma$ with multiplicity function $m_\Gamma:\Gamma\to\mathbb{Z}^+$ whose incompatibility graph is connected. Here the definition of the incompatibility graph is extended to multisets in the natural way. Let $\hat{\mathcal{G}}_C$ denote the collection of all multiset clusters of polymers from $\mathcal{C}$. Note that, for a given multiset $(\Gamma, m_\Gamma)$, there are precisely $\frac{\left(\sum_{\gamma\in\Gamma}m_\Gamma(\gamma)\right)!}{\prod_{\gamma\in\Gamma}m_\Gamma(\gamma)!}$ tuples that correspond to it. The abstract cluster expansion may then be written as
\begin{equation}
    \log(Z(\mathcal{C},w)) = \sum_{(\Gamma, m_\Gamma)\in\hat{\mathcal{G}}_C}\hat{\varphi}\left(H_{(\Gamma, m_\Gamma)}\right)\prod_{\gamma\in\Gamma}\frac{w_\gamma^{m_\Gamma(\gamma)}}{m_\Gamma(\gamma)!}, \notag
\end{equation}
where
\begin{equation}
    \hat{\varphi}(H) \coloneqq \sum_{\substack{S \subseteq E(H) \\ \text{spanning} \\ \text{connected}}}(-1)^{\abs{S}}. \notag
\end{equation}

\subsection{Approximation Schemes}
\label{section:ApproximationSchemes}

Let $\epsilon>0$ be a real number. An \emph{additive $\epsilon$-approximation} to $z$ is a complex number $\hat{z}$ such that $\abs{z-\hat{z}}\leq\epsilon$. A \emph{multiplicative $\epsilon$-approximation} to $z$ is a complex number $\hat{z}$ such that $\abs{z-\hat{z}}\leq\epsilon\abs{z}$. Note that an additive-error approximation to the logarithm of a number is equivalent to a multiplicative approximation to that number. A \emph{fully polynomial-time approximation scheme} for a sequence of complex numbers $(z_n)_{n\in\mathbb{N}}$ is a deterministic algorithm that, for any $n$ and $\epsilon>0$, produces a multiplicative $\epsilon$-approximation to $z_n$ in time polynomial in $n$ and $1/\epsilon$.

\subsection{Computational Complexity}
\label{section:ComputationalComplexity}

We shall refer to the following complexity classes: P (polynomial time), RP (randomised polynomial time), NP (non-deterministic polynomial time), and \#P. For a formal definition of these complexity classes,
we refer the reader to Ref.~\cite{arora2009computational}.

\section{General Framework}
\label{section:GeneralFramework}

\subsection{Approximation Algorithms}
\label{section:ApproximationAlgorithms}

In this section we establish a general framework for developing approximation algorithms for abstract polymer model partition functions. We consider abstract polymer models in which the polymers are connected subgraphs of bounded-degree bounded-rank multihypergraphs and compatibility is defined by vertex disjointness. When the polymer weights of these models decay sufficiently fast, then the logarithm of the partition function can be controlled by a convergent cluster expansion. Our algorithm approximates the logarithm of the partition function by computing the truncated cluster expansion to sufficiently high order.

Our general framework is based on that of Helmuth, Perkins, and Regts~\cite{helmuth2020algorithmic} and Borgs et al.~\cite{borgs2020efficient} where approximation algorithms were developed in the setting of bounded-degree graphs. Our algorithm can be viewed as a straightforward generalisation of theirs to the setting of bounded-degree bounded-rank multihypergraphs. Our main theorem is as follows.
\begin{theorem}
    \label{theorem:ApproximationAlgorithmAbstractPolymerModelPartitionFunction}
    Fix $\Delta,r\in\mathbb{Z}_{\geq2}$. Let $G=(V, E)$ be a multihypergraph of maximum degree at most $\Delta$ and rank at most $r$. Further let $(\mathcal{C}, w, \sim)$ be an abstract polymer model such that the polymers are connected subgraphs of $G$ and that two polymers $\gamma$ and $\gamma'$ are compatible if and only if $V(\gamma) \cap V(\gamma')=\varnothing$. Suppose that, for all polymers $\gamma\in\mathcal{C}$, the weight $w_\gamma$ can be computed in time $\exp(O(\norm{\gamma}))$ and satisfies
    \begin{equation}
        \abs{w_\gamma} \leq \left(\frac{1}{e^3\Delta\binom{r}{2}}\right)^\norm{\gamma}. \notag
    \end{equation}
    Then the cluster expansion for $\log(Z(\mathcal{C},w))$ converges absolutely, $Z(\mathcal{C},w)\neq0$, and there is a fully polynomial-time approximation scheme for $Z(\mathcal{C},w)$.
\end{theorem}

In Section~\ref{section:Applications} we shall apply Theorem~\ref{theorem:ApproximationAlgorithmAbstractPolymerModelPartitionFunction} to establish efficient approximation algorithms for several quantum problems.

Our proof of Theorem~\ref{theorem:ApproximationAlgorithmAbstractPolymerModelPartitionFunction} requires several lemmas. We first prove the following lemma which bounds the number of polymers of a certain size containing a particular vertex.

\begin{lemma}
    \label{lemma:ConnectedSubraphCount}
    Let $G=(V, E)$ be a multihypergraph of maximum degree at most $\Delta$ and rank at most $r$, and let $v \in V$ be a vertex. The number of connected subgraphs with $m$ edges that contain vertex $v$ is at most $\frac{(e\Delta(r-1))^m}{2}$.
\end{lemma}
\begin{proof}
    Let $C_{m,v}(G)$ denote the set of connected subgraphs of $G$ with $m$ edges that contain the vertex $v \in V$. Further let $T_{\Delta, r, v}$ denote the infinite $\Delta$-regular $r$-uniform linear hypertree with root $v$. Recall that a hypergraph is linear if the intersection of any pair of edges contains at most one vertex. Let $T^\star_{\Delta, r, v}$ be the graph with vertex set $\{v\} \cup E(T_{\Delta, r, v})$ and edges between vertices $v$ and $\varepsilon \in E(T_{\Delta, r, v})$ if and only if $v \in \varepsilon$ and edges between vertices $\varepsilon, \varepsilon' \in E(T_{\Delta, r, v})$ if and only if $\varepsilon \cap \varepsilon'\neq\varnothing$ and $d(v, \varepsilon) \neq d(v, \varepsilon')$. Note that $T^\star_{\Delta, r, v}$ is a tree with maximum degree precisely $(\Delta-1)(r-1)+1 \leq \Delta(r-1)$ and there is a natural bijection between $C_{m,v}(T_{\Delta, r, v})$ and $C_{m,v}(T^\star_{\Delta, r, v})$. The cardinality of $C_{m,v}(T^\star_{\Delta, r, v})$ is at most $\frac{1}{m+1}\binom{(m+1)\Delta(r-1)}{m}$~\cite[Lemma 2.1]{borgs2013left}.
    Hence, we have
    \begin{equation}
        \abs{C_{m,v}(G)} \leq \abs{C_{m,v}(T_{\Delta, r, v})} = \abs{C_{m,v}(T^\star_{\Delta, r, v})} \leq \frac{1}{m+1}\binom{(m+1)\Delta(r-1)}{m} \leq \frac{(e\Delta(r-1))^m}{2}, \notag
    \end{equation}
    completing the proof.
\end{proof}

\begin{remark}
    The proof of Lemma~\ref{lemma:ConnectedSubraphCount} gives a slightly sharper bound of $\frac{(e((\Delta-1)(r-1)+1))^m}{2}$. Improved bounds may be obtained for certain classes of multihypergraphs.
\end{remark}

We now show that provided the polymer weights decay sufficiently fast, then the cluster expansion converges absolutely and the truncated cluster expansion provides a good approximation to $\log(Z(\mathcal{C},w))$. This is formalised by the following lemma which utilises the Koteck\'y-Preiss convergence criterion.
\begin{lemma}
    \label{lemma:ConvergenceClusterExpansion}
    Let $G=(V, E)$ be a multihypergraph of maximum degree at most $\Delta$ and rank at most $r$. Further let $(\mathcal{C}, w, \sim)$ be an abstract polymer model such that the polymers are connected subgraphs of $G$ and that two polymers $\gamma$ and $\gamma'$ are compatible if and only if $V(\gamma) \cap V(\gamma')=\varnothing$. Suppose that, for all polymers $\gamma\in\mathcal{C}$, the weight $w_\gamma$ satisfies
    \begin{equation}
        \abs{w_\gamma} \leq \left(\frac{1}{e^3\Delta\binom{r}{2}}\right)^\norm{\gamma}. \notag
    \end{equation}
    Then the cluster expansion for $\log(Z(\mathcal{C},w))$ converges absolutely, $Z(\mathcal{C},w)\neq0$, and for $m\in\mathbb{Z}^+$,
    \begin{equation}
        \abs{T_m(Z(\mathcal{C},w))-\log(Z(\mathcal{C},w))} \leq \abs{G}e^{-\frac{m}{2}}. \notag
    \end{equation}
\end{lemma}
\begin{proof}
    We introduce a polymer $\gamma_v$ to every vertex $v$ in $G$ consisting of only that vertex. We define $\gamma_v$ to be incompatible with every polymer that contains $v$. Then, we have
    \begin{equation}
        \sum_{\gamma\nsim\gamma_v}\abs{w_\gamma}e^{\frac{1}{2}\left(\frac{1}{r-1}\abs{\gamma}+\norm{\gamma}\right)} \leq e^{\frac{1}{2(r-1)}}\sum_{\gamma\nsim\gamma_v}\abs{w_\gamma}e^{\norm{\gamma}} \leq e^{\frac{1}{2(r-1)}}\sum_{\gamma\nsim\gamma_v}\left(\frac{1}{e^2\Delta\binom{r}{2}}\right)^\norm{\gamma}, \notag
    \end{equation}
    where we have used the fact that $\abs{\gamma}\leq(r-1)\norm{\gamma}+1$. By Lemma~\ref{lemma:ConnectedSubraphCount}, the number of polymers $\gamma$ with $\norm{\gamma}=m$ that are incompatible with $\gamma_v$ is at most $\frac{(e\Delta(r-1))^m}{2}$. Thus, we may write
    \begin{equation}
        \sum_{\gamma\nsim\gamma_v}\abs{w_\gamma}e^{\frac{1}{2}\left(\frac{1}{r-1}\abs{\gamma}+\norm{\gamma}\right)} \leq \frac{e^{\frac{1}{2(r-1)}}}{2}\sum_{m=1}^\infty\left(\frac{2}{er}\right)^m \leq \frac{1}{2(r-1)}. \notag
    \end{equation}
    Fix a polymer $\gamma$. By summing over all vertices in $\gamma$, we obtain
    \begin{equation}
        \sum_{\gamma^*\nsim\gamma}\abs{w_{\gamma^*}}e^{\frac{1}{2}\left(\frac{1}{r-1}\abs{\gamma^*}+\norm{\gamma^*}\right)} \leq \frac{1}{2(r-1)}\abs{\gamma}. \notag
    \end{equation}
    Now by applying Theorem~\ref{theorem:KoteckyPreissConvergenceCriterion} with $a(\gamma)=\frac{1}{2(r-1)}\abs{\gamma}$ and $d(\gamma)=\frac{1}{2}\norm{\gamma}$, we have that the cluster expansion converges absolutely, $Z(\mathcal{C},w)\neq0$, and
    \begin{equation}
        \sum_{\substack{\Gamma\in\mathcal{G}_C \\ \Gamma\ni\gamma_v }}\abs{\varphi(H_\Gamma)\prod_{\gamma\in\Gamma}w_\gamma}e^{\frac{1}{2}\norm{\Gamma}} \leq 1. \notag
    \end{equation}
    By summing over all vertices in $G$, we obtain
    \begin{equation}
        \sum_{\substack{\Gamma\in\mathcal{G}_C \\ \norm{\Gamma} \geq m }}\abs{\varphi(H_\Gamma)\prod_{\gamma\in\Gamma}w_\gamma} \leq \abs{G}e^{-\frac{m}{2}}, \notag
    \end{equation}
    completing the proof.
\end{proof}

Lemma~\ref{lemma:ConvergenceClusterExpansion} implies that to obtain a multiplicative $\epsilon$-approximation $Z(\mathcal{C},w)$, it is sufficient to compute the truncated cluster expansion $T_m(Z(\mathcal{C},w))$ to order $m=O(\log(\abs{G}/\epsilon))$. We shall now establish an algorithm for computing $T_m(Z(\mathcal{C},w))$ in time $\exp(O(m))\cdot\abs{G}^{O(1)}$. This requires the following two lemmas.
\begin{lemma}
    \label{lemma:ListClustersAlgorithm}
    Fix $\Delta,r\in\mathbb{Z}_{\geq2}$. Let $G=(V, E)$ be a multihypergraph of maximum degree at most $\Delta$ and rank at most $r$. Further let $(\mathcal{C}, w, \sim)$ be an abstract polymer model such that the polymers are connected subgraphs of $G$ and that two polymers $\gamma$ and $\gamma'$ are compatible if and only if $V(\gamma) \cap V(\gamma')=\varnothing$. The clusters of size at most $m$ can be listed in time $\exp(O(m))\cdot\abs{G}^{O(1)}$.
\end{lemma}
\begin{proof}
    Our proof follows a similar approach to that of Ref.~\cite[Theorem 6]{helmuth2020algorithmic}. We list all connected subgraphs of $G$ with at most $m$ edges in time $\exp(O(m))\cdot\abs{G}^{O(1)}$ by depth-first search. For each of these subgraphs, we consider all ways to label the edges with positive integers such that their sum is at most $m$ in time $\exp(O(m))$. For each of these labelled subgraphs, we consider all clusters that correspond to it, i.e., clusters whose multiset sum over polymers induces the subgraph with multiplicities given by the edge labels.

    We now prove by induction that the number of such clusters for a subgraph with label sum $m$ is at most $(e\Delta(r-1))^{2m}$. This is clearly true when $m=0$. Now suppose that the number of such clusters for a subgraph with label sum $m$ is at most $(e\Delta(r-1))^{2m}$. For a subgraph with label sum $m+1$, we choose an arbitrary vertex in the subgraph and consider all polymers that contain that vertex. By Lemma~\ref{lemma:ConnectedSubraphCount}, there are at most $(e\Delta(r-1))^n$ such polymers of size $n$. By removing each polymer from the subgraph and applying the induction hypothesis, we have that the number of clusters in the subgraph is at most
    \begin{equation}
        \sum_{n=1}^{m+1}(e\Delta(r-1))^n(e\Delta(r-1))^{2(m+1-n)} \leq (e\Delta(r-1))^{2(m+1)}\sum_{n=1}^{m+1}(e\Delta(r-1))^{-n} \leq (e\Delta(r-1))^{2(m+1)}, \notag
    \end{equation}
    completing the induction. These clusters can be enumerated in time $\exp(O(m))$ by depth-first search, completing the proof.
\end{proof}

\begin{lemma}
    \label{lemma:UrsellFunctionAlgorithm}
    The Ursell function $\varphi(H)$ can be computed in time $\exp(O(\abs{H}))$.
\end{lemma}
\begin{proof}
    Our proof follows that of Ref.~\cite[Lemma~5]{helmuth2020algorithmic}. For a connected graph $H$, we have
    \begin{equation}
        \varphi(H) = \frac{1}{\norm{H}!}\sum_{\substack{S \subseteq E(H) \\ \text{spanning} \\ \text{connected}}}(-1)^{\abs{S}} = -\frac{(-1)^{\abs{H}}}{\norm{H}!}T_H(0,1), \notag
    \end{equation}
    where $T_H(x,y)$ denotes the Tutte polynomial of $H$ defined by
    \begin{equation}
        T_H(x,y) \coloneqq \sum_{S \subseteq E}(x-1)^{k(S)-k(E)}(y-1)^{k(S)+\abs{S}-\abs{H}}. \notag
    \end{equation}
    Here $k(S)$ denotes the number of connected components of the subgraph with edge set $S$. The Ursell function can then be computed in time $\exp(O(\abs{H}))$ by evaluating the Tutte polynomial in time $\exp(O(\abs{H}))$ via an algorithm of B\"orklund et al.~\cite{bjorklund2008computing}. This completes the proof.
\end{proof}

\begin{lemma}
    \label{lemma:TruncatedClusterExpansionApproximationAlgorithm}
    Fix $\Delta,r\in\mathbb{Z}_{\geq2}$. Let $G=(V, E)$ be a multihypergraph of maximum degree at most $\Delta$ and rank at most $r$. Further let $(\mathcal{C}, w, \sim)$ be an abstract polymer model such that the polymers are connected subgraphs of $G$ and that two polymers $\gamma$ and $\gamma'$ are compatible if and only if $V(\gamma) \cap V(\gamma')=\varnothing$. Suppose that, for all polymers $\gamma\in\mathcal{C}$, the weight $w_\gamma$ can be computed in time $\exp(O(\norm{\gamma}))$. Then the truncated cluster expansion $T_m(Z(\mathcal{C},w))$ can be computed in time $\exp(O(m))\cdot\abs{G}^{O(1)}$.
\end{lemma}
\begin{proof}
    We can list all clusters of size at most $m$ in time $\exp(O(m))\cdot\abs{G}^{O(1)}$ by Lemma~\ref{lemma:ListClustersAlgorithm}. For each of these clusters, we can compute the Ursell function in time $\exp(O(m))$ by Lemma~\ref{lemma:UrsellFunctionAlgorithm}, and the polymer weights in time $\exp(O(m))$ by assumption. Hence, the truncated cluster expansion $T_m(Z(\mathcal{C},w))$ can be computed in time $\exp(O(m))\cdot\abs{G}^{O(1)}$.
\end{proof}

Combining Lemma~\ref{lemma:ConvergenceClusterExpansion} with Lemma~\ref{lemma:TruncatedClusterExpansionApproximationAlgorithm} proves Theorem~\ref{theorem:ApproximationAlgorithmAbstractPolymerModelPartitionFunction}.

\subsection{Hardness of Approximation}
\label{section:HardnessOfApproximation}

In this section we establish the hardness of approximating abstract polymer model partition functions. In particular, we establish the hardness of approximating the Ising model partition function at imaginary temperature on bounded-degree graphs, which will be useful for our purposes via reductions. This
setting was studied in Ref.~\cite{galanis2022complexity}, which established hardness of approximation results for this problem. We utilise the results of Ref.~\cite{galanis2022complexity} to obtain significantly sharper bounds when the maximum degree is sufficiently large.

We model an Ising system by a multigraph $G=(V, E)$. At each vertex $v$ of $G$ there is a two-dimensional classical spin space $\{-1,+1\}$. The classical spin space on the multihypergraph is given by $\{-1,+1\}^V$. An interaction $\phi$ assigns a real number $\phi(\varepsilon)$ to each edge $\varepsilon$ of $G$. We are interested in the partition function $Z_{\mathrm{Ising}}(G;\beta)$ at inverse temperature $\beta$, defined by
\begin{equation}
    Z_{\mathrm{Ising}}(G;\beta) \coloneqq \sum_{\sigma\in\{-1,+1\}^V}\prod_{\{u,v\} \in E}e^{-\beta\phi\left(\{u,v\}\right)\sigma_u\sigma_v}. \notag
\end{equation}
We shall normalise the partition function by a multiplicative factor of $\frac{1}{2^\abs{G}}$. Further, we shall assume that $\abs{\phi(\varepsilon)}\leq1$ for all $\varepsilon \in E$, which is always possible by a rescaling of $\beta$. We shall consider the case where the inverse temperature $\beta$ is imaginary, i.e., $\beta=i\theta$ for $\theta\in\mathbb{R}$. The Ising model partition function at imaginary temperature arises naturally in the probability amplitudes of quantum circuits~\cite{mann2019approximation}. Our hardness result concerning the approximation of $Z_{\mathrm{Ising}}(G;i\theta)$ is as follows.
\begin{theorem}
    \label{theorem:IsingModelPartitionFunctionApproximationHardness}
    Fix $\epsilon>0$, $\Delta\in\mathbb{Z}_{\geq3}$, and $\theta\in\mathbb{R}$ such that $\abs{\theta}\geq\frac{3\pi}{5(\Delta-2)}$. It is \emph{\#P-hard} to approximate the Ising model partition function $Z_{\mathrm{Ising}}(G;i\theta)$ up to a multiplicative $\epsilon$-approximation on multigraphs of maximum degree at most $\Delta$.
\end{theorem}
\begin{proof}
    By Ref.~\cite[Theorem 3]{galanis2022complexity}, it is \#P-hard to approximate the Ising model partition function $Z_{\mathrm{Ising}}(G;i\theta)$ up to a multiplicative $\epsilon$-approximation on multigraphs of maximum degree $3$ for $\abs{\theta}\geq\frac{\pi}{5}>\arctan\left(\frac{1}{\sqrt{2}}\right)$. For a graph $G$ of maximum degree $3$ and a positive integer $k\in\mathbb{Z}^+$, let $G_k$ denote the \emph{$k$-thickening} of $G$, that is, the multigraph formed by replacing each edge of $G$ with $k$ parallel edges. Note that the maximum degree of $G_k$ is precisely $3k$. Now observe that, for any $k\in\mathbb{Z}^+$, we have $Z_{\mathrm{Ising}}(G;i\theta)=Z_{\mathrm{Ising}}\left(G_k;\frac{i\theta}{k}\right)$. Hence, it is \#P-hard to approximate $Z_{\mathrm{Ising}}(G;i\theta)$ up to a multiplicative $\epsilon$-approximation on multigraphs of maximum degree at most $3k$ for $\theta\geq\frac{\pi}{5k}$. It follows that it is \#P-hard to approximate $Z_{\mathrm{Ising}}(G;i\theta)$ up to a multiplicative $\epsilon$-approximation on multigraphs of maximum degree at most $\Delta$ for $\theta\geq\frac{3\pi}{5(\Delta-2)}$, completing the proof.
\end{proof}

\begin{remark}
    The proof of Theorem~\ref{theorem:IsingModelPartitionFunctionApproximationHardness} gives a slightly sharper bound. Further, the proof technique may be applied to the case of complex $\beta$.
\end{remark}

This offers a significant improvement over Ref.~\cite{galanis2022complexity} when $\Delta\geq7$, which applies when $\abs{\theta}>\arctan\left(\frac{1}{\sqrt{\Delta-1}}\right)$. In Section~\ref{section:Applications} we shall apply Theorem~\ref{theorem:ApproximationAlgorithmAbstractPolymerModelPartitionFunction} to establish the hardness of approximation of several quantum problems. We shall now show that the Ising model partition function $Z_{\mathrm{Ising}}(G;\beta)$ admits an abstract polymer model representation. This is formalised by the following lemma.
\begin{lemma}
    \label{lemma:IsingModelPartitionFunctionAbstractPolymerModel}
    The Ising model partition function $Z_{\mathrm{Ising}}(G;\beta)$ admits the following abstract polymer model representation.
    \begin{equation}
        Z_{\mathrm{Ising}}(G;\beta) = \sum_{\Gamma\in\mathcal{G}}\prod_{\gamma\in\Gamma}w_\gamma, \notag
    \end{equation}
    where
    \begin{equation}
        w_\gamma \coloneqq \frac{1}{2^\abs{\gamma}}\sum_{\sigma\in\{-1,+1\}^{V(\gamma)}}\prod_{\{u,v\} \in E(\gamma)}\left(e^{-\beta\phi\left(\{u,v\}\right)\sigma_u\sigma_v}-1\right). \notag
    \end{equation}
\end{lemma}
\begin{proof}
    By applying Lemma~\ref{lemma:InclusionExclusion} with $f(E)=\frac{1}{2^\abs{G}}\sum_{\sigma\in\{-1,+1\}^V}\prod_{\{u,v\} \in E}e^{-\beta\phi\left(\{u,v\}\right)\sigma_u\sigma_v}$, we have
    \begin{align}
        Z_{\mathrm{Ising}}(G;\beta) &= \frac{1}{2^\abs{G}}\sum_{\sigma\in\{-1,+1\}^V}\prod_{\{u,v\} \in E}e^{-\beta\phi\left(\{u,v\}\right)\sigma_u\sigma_v} \notag \\
        &= \frac{1}{2^\abs{G}}\sum_{S \subseteq E}(-1)^\abs{S}\sum_{T \subseteq S}(-1)^\abs{T}\sum_{\sigma\in\{-1,+1\}^V}\prod_{\{u,v\} \in T}e^{-\beta\phi\left(\{u,v\}\right)\sigma_u\sigma_v}. \notag
    \end{align}
    For a subset $S \subseteq E$, let $\Gamma_S$ denote the maximally connected components of $S$. Recall that polymers are connected subgraphs of $G$ with compatibility relations defined by vertex disjointness. For a polymer $\gamma$, $\abs{\gamma}$ denotes the order and $\norm{\gamma}$ denotes the size of the corresponding subgraph. By factorising over these components, we have
    \begin{align}
        Z_{\mathrm{Ising}}(G;\beta) &= \sum_{S \subseteq E}\prod_{\gamma \in \Gamma_S}(-1)^\norm{\gamma}\sum_{T \subseteq E(\gamma)}(-1)^\abs{T}\frac{1}{2^\abs{\gamma}}\sum_{\sigma\in\{-1,+1\}^{V(\gamma)}}\prod_{\{u,v\} \in T}e^{-\beta\phi\left(\{u,v\}\right)\sigma_u\sigma_v} \notag \\
        &= \sum_{S \subseteq E}\prod_{\gamma \in \Gamma_S}\frac{1}{2^\abs{\gamma}}\sum_{\sigma\in\{-1,+1\}^{V(\gamma)}}\prod_{\{u,v\} \in E(\gamma)}\left(e^{-\beta\phi\left(\{u,v\}\right)\sigma_u\sigma_v}-1\right) \notag \\
        &= \sum_{\Gamma\in\mathcal{G}}\prod_{\gamma\in\Gamma}w_\gamma. \notag
    \end{align}
    This completes the proof.
\end{proof}

We note that Lemma~\ref{lemma:IsingModelPartitionFunctionAbstractPolymerModel} can be combined with Theorem~\ref{theorem:ApproximationAlgorithmAbstractPolymerModelPartitionFunction} to establish an efficient approximation algorithm for $Z_{\mathrm{Ising}}(G;\beta)$ on graphs of maximum degree at most $\Delta$ when $\abs{\beta}\leq\frac{1}{e^4\Delta}$. Efficient approximation algorithms with significantly sharper bounds have previously been established~\cite{mann2019approximation, galanis2022complexity}. In particular, Ref.~\cite{galanis2022complexity} established an efficient approximation algorithm that applies when $\abs{\beta}<\frac{\pi}{4(\Delta-1)}$. In the case when $\beta$ is real, the exact point of a computational complexity transition is known under the complexity-theoretic assumption that RP is not equal to NP due to the approximation algorithm of Ref.~\cite{sinclair2014approximation} and the hardness of approximation results of Refs.~\cite{sly2012computational, galanis2016inapproximability}.

\section{Applications}
\label{section:Applications}

In this section we apply our algorithmic framework to establish efficient approximation algorithms for classes of quantum problems. This includes probability amplitudes, expectation values, partition functions, and thermal expectation values. We apply our hardness of approximation framework to show the optimality of our algorithmic conditions for probability amplitudes and partition functions under complexity-theoretic assumptions. Further, we show that our algorithmic condition is almost optimal for expectation values and optimal for thermal expectation values in the sense of zero freeness.

\subsection{Probability Amplitudes}
\label{section:ProbabilityAmplitudes}

In this section we study the problem of approximating probability amplitudes of quantum circuits. This problem is known to be \#P-hard in general~\cite{goldberg2017complexity}; however, we show that, for a class of quantum circuits close to the identity, this problem admits an efficient approximation algorithm. Further, we show that this algorithmic condition is optimal under complexity-theoretic assumptions.

We model a quantum circuit by a multihypergraph $G=(V, E)$. At each vertex $v$ of $G$ there is a $d$-dimensional Hilbert space $\mathcal{H}_v$ with $d<\infty$. The Hilbert space on the multihypergraph is given by $\mathcal{H}_G\coloneqq\bigotimes_{v \in V}\mathcal{H}_v$. An interaction $U$ assigns a unitary operator $U_\varepsilon$ on $\mathcal{H}_\varepsilon\coloneqq\bigotimes_{v \in \varepsilon}\mathcal{H}_v$ to each edge $\varepsilon$ of $G$. We shall assume there is an implicit ordering of the unitary operators given by the edge labels which determines the order in which products of these operators are taken. The quantum circuit on $G$ is defined by $U_G\coloneqq\prod_{\varepsilon \in E}U_\varepsilon$. We are interested in the probability amplitude $A_{U_G}$, defined by $A_{U_G}\coloneqq\matrixel{0^\abs{G}}{U_G}{0^\abs{G}}$. Note that any probability amplitude may be expressed in this form by a simple modification of the circuit. Our algorithmic result concerning the approximation of $A_{U_G}$ is as follows.
\begin{theorem}
    \label{theorem:ApproximationAlgorithmQuantumProbabilityAmplitude}
    Fix $\Delta,r\in\mathbb{Z}_{\geq2}$. Let $G=(V, E)$ be a multihypergraph of maximum degree at most $\Delta$ and rank at most $r$. Suppose that, for all $\varepsilon \in E$,
    \begin{equation}
        \norm{U_\varepsilon-\mathbb{I}} \leq \frac{1}{e^3\Delta\binom{r}{2}}. \notag
    \end{equation}
    Then the cluster expansion for $\log(A_{U_G})$ converges absolutely, $A_{U_G}\neq0$, and there is a fully polynomial-time approximation scheme for $A_{U_G}$.
\end{theorem}

\begin{remark}
    Theorem~\ref{theorem:ApproximationAlgorithmQuantumProbabilityAmplitude} also applies to probability amplitudes of the form $\matrixel{\psi}{U_G}{\psi}$, where $\ket{\psi}$ is a product state over qudits, i.e., $\ket{\psi}\coloneqq\bigotimes_{v \in V}\ket{\psi_v}$. Further, Theorem~\ref{theorem:ApproximationAlgorithmQuantumProbabilityAmplitude} applies to unitary operators of the form $U_\varepsilon=e^{-i\theta\Phi(\varepsilon)}$, where $\theta$ is a real number such that $\abs{\theta}\leq\frac{1}{e^4\Delta\binom{r}{2}}$ and $\Phi(\varepsilon)$ is a self-adjoint operator on $\mathcal{H}_\varepsilon$ with $\norm{\Phi(\varepsilon)}\leq1$.
\end{remark}

We prove Theorem~\ref{theorem:ApproximationAlgorithmQuantumProbabilityAmplitude} by showing that the conditions required to apply Theorem~\ref{theorem:ApproximationAlgorithmAbstractPolymerModelPartitionFunction} are satisfied. That is, we show that (1) the probability amplitude $A_{U_G}$ admits a suitable abstract polymer model representation, (2) the polymer weights satisfy the desired bound, and (3) the polymer weights can be computed in the desired time. This is achieved in the following three lemmas.
\begin{lemma}
    \label{lemma:QuantumProbabilityAmplitudeAbstractPolymerModel}
    The probability amplitude $A_{U_G}$ admits the following abstract polymer model representation.
    \begin{equation}
        A_{U_G} = \sum_{\Gamma\in\mathcal{G}}\prod_{\gamma\in\Gamma}w_\gamma, \notag
    \end{equation}
    where
    \begin{equation}
        w_\gamma \coloneqq \matrixel{0^\abs{\gamma}}{\left[\prod_{\varepsilon \in E(\gamma)}(U_\varepsilon-\mathbb{I})\right]}{0^\abs{\gamma}}. \notag
    \end{equation}
\end{lemma}
\begin{proof}
    By applying Lemma~\ref{lemma:InclusionExclusion} with $f(E)=\matrixel{0^\abs{G}}{\left(\prod_{\varepsilon \in E}U_\varepsilon\right)}{0^\abs{G}}$, we have
    \begin{align}
        A_{U_G} &= \matrixel{0^\abs{G}}{U_G}{0^\abs{G}} \notag \\
        &= \sum_{S \subseteq E}(-1)^\abs{S}\sum_{T \subseteq S}(-1)^\abs{T}\matrixel{0^\abs{G}}{\left(\prod_{\varepsilon \in T}U_\varepsilon\right)}{0^\abs{G}}. \notag
    \end{align}
    For a subset $S \subseteq E$, let $\Gamma_S$ denote the maximally connected components of $S$. By factorising over these components, we have
    \begin{align}
        A_{U_G} &= \sum_{S \subseteq E}\prod_{\gamma \in \Gamma_S}(-1)^\norm{\gamma}\sum_{T \subseteq E(\gamma)}(-1)^\abs{T}\matrixel{0^\abs{\gamma}}{\left(\prod_{\varepsilon \in T}U_\varepsilon\right)}{0^\abs{\gamma}} \notag \\
        &= \sum_{S \subseteq E}\prod_{\gamma \in \Gamma_S}\matrixel{0^\abs{\gamma}}{\left[\prod_{\varepsilon \in E(\gamma)}(U_\varepsilon-\mathbb{I})\right]}{0^\abs{\gamma}} \notag \\
        &= \sum_{\Gamma\in\mathcal{G}}\prod_{\gamma\in\Gamma}w_\gamma. \notag
    \end{align}
    This completes the proof.
\end{proof}

\begin{lemma}
    \label{lemma:QuantumProbabilityAmplitudeWeightBound}
    Fix $\Delta,r\in\mathbb{Z}_{\geq2}$. Let $G=(V, E)$ be a multihypergraph of maximum degree at most $\Delta$ and rank at most $r$. Suppose that, for all $\varepsilon \in E$,
    \begin{equation}
        \norm{U_\varepsilon-\mathbb{I}} \leq \frac{1}{e^3\Delta\binom{r}{2}}. \notag
    \end{equation}
    Then, for all polymers $\gamma\in\mathcal{C}$, the weight $w_\gamma$ satisfies
    \begin{equation}
        \abs{w_\gamma} \leq \left(\frac{1}{e^3\Delta\binom{r}{2}}\right)^\norm{\gamma}. \notag
    \end{equation}
\end{lemma}
\begin{proof}
    Fix a polymer $\gamma$. We have
    \begin{equation}
        \abs{w_\gamma} \leq \prod_{\varepsilon \in E(\gamma)}\norm{U_\varepsilon-\mathbb{I}} \leq \left(\frac{1}{e^3\Delta\binom{r}{2}}\right)^\norm{\gamma}, \notag
    \end{equation}
    completing the proof.
\end{proof}

\begin{lemma}
    \label{lemma:QuantumProbabilityAmplitudeWeightAlgorithm}
    The weight $w_\gamma$ of a polymer $\gamma$ can be computed in time $\exp(O(\norm{\gamma}))$.
\end{lemma}
\begin{proof}
    The result follows by sparse matrix-vector multiplication.
\end{proof}

Combining Theorem~\ref{theorem:ApproximationAlgorithmAbstractPolymerModelPartitionFunction} with Lemma~\ref{lemma:QuantumProbabilityAmplitudeAbstractPolymerModel}, Lemma~\ref{lemma:QuantumProbabilityAmplitudeWeightBound}, and Lemma~\ref{lemma:QuantumProbabilityAmplitudeWeightAlgorithm} proves Theorem~\ref{theorem:ApproximationAlgorithmQuantumProbabilityAmplitude}. We now show that the algorithmic condition of Theorem~\ref{theorem:ApproximationAlgorithmQuantumProbabilityAmplitude} is optimal in the case of multigraphs under complexity-theoretic assumptions. This is achieved by establishing a hardness of approximation result for the probability amplitude $A_{U_G}$. For convenience, we shall consider unitary operators of the form $U_\varepsilon=e^{-i\theta\Phi(\varepsilon)}$, where $\theta$ is a real number and $\Phi(\varepsilon)$ is a self-adjoint operator on $\mathcal{H}_\varepsilon$ with $\norm{\Phi(\varepsilon)}\leq1$. Our hardness result concerning the approximation of $A_{U_G}(\theta)$ is as follows.
\begin{theorem}
    \label{theorem:ApproximationHardnessQuantumProbabilityAmplitude}
    Fix $\epsilon>0$, $\Delta\in\mathbb{Z}_{\geq3}$, and $\theta\in\mathbb{R}$ such that $\abs{\theta}\geq\frac{3\pi}{5(\Delta-2)}$. It is \emph{\#P-hard} to approximate the probability amplitude $A_{U_G}(\theta)$ up to a multiplicative $\epsilon$-approximation on multigraphs of maximum degree at most $\Delta$.
\end{theorem}
\begin{proof}
    Our proof is based on a reduction from the Ising model partition function. We consider quantum circuits on multigraphs with a two-dimensional Hilbert space at each vertex and unitary operators of the form $U_\varepsilon=e^{-i\theta\phi(\varepsilon)\bigotimes_{v \in \varepsilon}X_v}$, where $\phi(\varepsilon)$ is a real number satisfying $\abs{\phi(\varepsilon)}\leq1$. We have
    \begin{align}
        A_{U_G}(\theta) &= \matrixel{0^\abs{G}}{\left(\prod_{\varepsilon \in E}e^{-i\theta\phi(\varepsilon)\bigotimes_{v \in \varepsilon}X_v}\right)}{0^\abs{G}} \notag \\
        &= \frac{1}{2^\abs{G}}\sum_{\sigma\in\{-1,+1\}^V}\prod_{\{u,v\} \in E}e^{-i\theta\phi\left(\{u,v\}\right)\sigma_u\sigma_v} \notag \\
        &= \frac{1}{2^\abs{G}}Z_{\mathrm{Ising}}(G;i\theta). \notag
    \end{align}
    The proof then follows from Theorem~\ref{theorem:IsingModelPartitionFunctionApproximationHardness}.
\end{proof}

Our results establish a 
computational complexity transition from P to \#P-hard for the problem of approximating probability amplitudes.

\subsection{Expectation Values}
\label{section:ExpectationValues}

In this section we study the problem of approximating expectation values of quantum circuits. This problem is known to be \#P-hard in general~\cite{goldberg2017complexity}; in particular, it is a special case of computing output probabilities of quantum circuits. We show that, for a class of quantum circuits with operators close to the identity, this problem admits an efficient approximation algorithm. This setting was studied in Ref.~\cite{bravyi2021classical}, which established an efficient approximation algorithm for this problem. Our approach offers a simpler and sharper analysis in a more slightly general setting. Further, we show that this algorithmic condition is almost optimal in the sense of zero freeness.

We model a quantum circuit by a multihypergraph $G=(V, E)$ as in Section~\ref{section:ProbabilityAmplitudes} and assume that the size of $G$ is at most a polynomial in the order of $G$. An operator $O$ assigns a self-adjoint operator $O_v$ on $\mathcal{H}_v$ to each vertex $v$ of $G$. The operator $O_G$ on $G$ is defined by $O_G\coloneqq\prod_{v \in V}O_v$. We are interested in the expectation value $\expval{O}_{U_G}$, defined by $\expval{O}_{U_G}\coloneqq\matrixel{0^\abs{G}}{{U_G}^\dag O_G U_G}{0^\abs{G}}$.

We now introduce some further definitions that will be useful for our analysis. Let $S_E\coloneqq(\varepsilon)_{\varepsilon \in E}$ denote the sequence of edges from $G$ sorted in increasing order with respect to the edge labels. For a vertex $v$ of $G$, let $S_v$ denote the longest increasing subsequence of $S_E$ such that every prefix induces a connected subgraph of $G$ containing $v$. We define the \emph{causal subgraph} $C_v$ of $v$ to be the subgraph of $G$ induced by the sequence $S_v$. For a subset $U$ of vertices of $G$, we define the \emph{causal subgraph} $C_U$ of $U$ to be the subgraph of $G$ induced by the set $\bigcup_{v \in U}E(C_v)$. We define the \emph{causal intersection hypergraph} $C(G)$ of $G$ to be the hypergraph with vertex set $V$ and edge set $\{V(C_v)\}_{v \in V}$. We identify the edges of $C(G)$ with the vertices of $G$. Note that the connected components of a subgraph $S$ of $C(G)$ are in one-to-one correspondence with the connected components of $C_{E(S)}$. We shall consider polymers that are connected subgraphs of $C(G)$. Our algorithmic result concerning the approximation of $\expval{O}_{U_G}$ is as follows.

\begin{theorem}
    \label{theorem:ApproximationAlgorithmQuantumExpectationValue}
    Fix $\Delta,r\in\mathbb{Z}_{\geq2}$. Let $G=(V, E)$ be a multihypergraph such that the causal intersection hypergraph $C(G)$ of $G$ has maximum degree at most $\Delta$ and rank at most $r$. Suppose that, for all $v \in V$,
    \begin{equation}
        \norm{O_v-\mathbb{I}} \leq \frac{1}{e^3\Delta\binom{r}{2}}. \notag
    \end{equation}
    Then the cluster expansion for $\log(\expval{O}_{U_G})$ converges absolutely, $\expval{O}_{U_G}\neq0$, and there is a fully polynomial-time approximation scheme for $\expval{O}_{U_G}$.
\end{theorem}

\begin{remark}
    Theorem~\ref{theorem:ApproximationAlgorithmQuantumExpectationValue} may be extended to a slightly more general class of product operators.
\end{remark}

In the case when $G$ corresponds to a quantum circuit $U_G$ of depth at most $d$ with each gate acting on at most $k$ qudits, the causal intersection hypergraph $C(G)$ has maximum degree at most $k^d$ and rank at most $k^d$. Further, when $G$ is restricted to edges on the lattice graph $\mathbb{Z}^\nu$, the causal intersection hypergraph $C(G)$ has maximum degree at most $(2d)^\nu$ and rank at most $(2d)^\nu$. This implies that our algorithm may be applied to these classes of quantum circuits when $\norm{O_v-\mathbb{I}}\leq\frac{2}{e^3k^{3d}}$ and $\norm{O_v-\mathbb{I}}\leq\frac{2}{e^3(2d)^{3\nu}}$ for all $v \in V$, respectively. A more refined analysis in the latter case shows that our algorithm may be applied when $\norm{O_v-\mathbb{I}}\leq\frac{2}{e^32^{3\nu}d^{2\nu}}$  for all $v \in V$. This offers a significant improvement over Ref.~\cite{bravyi2021classical}, which applies to these classes when $\norm{O_v-\mathbb{I}}<\frac{1}{60k^{5d}}$ and $\norm{O_v-\mathbb{I}}<\frac{1}{60(16d)^{2\nu}}$ for all $v \in V$, respectively.

We prove Theorem~\ref{theorem:ApproximationAlgorithmQuantumExpectationValue} by showing that the conditions required to apply Theorem~\ref{theorem:ApproximationAlgorithmAbstractPolymerModelPartitionFunction} are satisfied. That is, we show that (1) the expectation value $\expval{O}_{U_G}$ admits a suitable abstract polymer model representation, (2) the polymer weights satisfy the desired bound, and (3) the polymer weights can be computed in the desired time. This is achieved in the following three lemmas.
\begin{lemma}
    \label{lemma:QuantumExpectationValueAbstractPolymerModel}
    The expectation value $\expval{O}_{U_G}$ admits the following abstract polymer model representation.
    \begin{equation}
        \expval{O}_{U_G} = \sum_{\Gamma\in\mathcal{G}}\prod_{\gamma\in\Gamma}w_\gamma, \notag
    \end{equation}
    where
    \begin{equation}
        w_\gamma \coloneqq \matrixel{0^\abs{\gamma}}{{U_{C_{E(\gamma)}}}^\dag\left[\prod_{\varepsilon \in E(\gamma)}(O_\varepsilon-\mathbb{I})\right]U_{C_{E(\gamma)}}}{0^\abs{\gamma}}. \notag
    \end{equation}
\end{lemma}
\begin{proof}
    By applying Lemma~\ref{lemma:InclusionExclusion} with $f(V)=\matrixel{0^\abs{G}}{{U_G}^\dag\left(\prod_{v \in V}O_v\right)U_G}{0^\abs{G}}$, we have
    \begin{align}
        \expval{O}_{U_G} &= \matrixel{0^\abs{G}}{{U_G}^\dag O_G U_G}{0^\abs{G}} \notag \\
        &= \sum_{S \subseteq V}(-1)^\abs{S}\sum_{T \subseteq S}(-1)^\abs{T}\matrixel{0^\abs{G}}{{U_G}^\dag\left(\prod_{v \in T}O_v\right)U_G}{0^\abs{G}} \notag \\
        &= \sum_{S \subseteq E(C(G))}(-1)^\abs{S}\sum_{T \subseteq S}(-1)^\abs{T}\matrixel{0^\abs{G}}{{U_G}^\dag\left(\prod_{\varepsilon \in T}O_\varepsilon\right)U_G}{0^\abs{G}}. \notag
    \end{align}
    For a subset $S \subseteq E(C(G))$, let $\Gamma_S$ denote the maximally connected components of $S$. By factorising over these components, we have
    \begin{align}
        \expval{O}_{U_G} &= \sum_{S \subseteq E}\prod_{\gamma \in \Gamma_S}(-1)^\norm{\gamma}\sum_{T \subseteq E(\gamma)}(-1)^\abs{T}\matrixel{0^\abs{\gamma}}{{U_G}^\dag\left(\prod_{\varepsilon \in T}O_\varepsilon\right)U_G}{0^\abs{\gamma}} \notag \\
        &= \sum_{S \subseteq E}\prod_{\gamma \in \Gamma_S}(-1)^\norm{\gamma}\sum_{T \subseteq E(\gamma)}(-1)^\abs{T}\matrixel{0^\abs{\gamma}}{{U_{C_{E(\gamma)}}}^\dag\left(\prod_{\varepsilon \in T}O_\varepsilon\right)U_{C_{E(\gamma)}}}{0^\abs{\gamma}} \notag \\
        &= \sum_{S \subseteq E}\prod_{\gamma \in \Gamma_S}\matrixel{0^\abs{\gamma}}{{U_{C_{E(\gamma)}}}^\dag\left[\prod_{\varepsilon \in E(\gamma)}(O_\varepsilon-\mathbb{I})\right]U_{C_{E(\gamma)}}}{0^\abs{\gamma}} \notag \\
        &= \sum_{\Gamma\in\mathcal{G}}\prod_{\gamma\in\Gamma}w_\gamma. \notag
    \end{align}
    This completes the proof.
\end{proof}

\begin{lemma}
    \label{lemma:QuantumExpectationValueWeightBound}
    Fix $\Delta,r\in\mathbb{Z}_{\geq2}$. Let $G=(V, E)$ be a multihypergraph such that the causal intersection hypergraph $C(G)$ of $G$ has maximum degree at most $\Delta$ and rank at most $r$. Suppose that, for all $v \in V$,
    \begin{equation}
        \norm{O_v-\mathbb{I}} \leq \frac{1}{e^3\Delta\binom{r}{2}}. \notag
    \end{equation}
    Then, for all polymers $\gamma\in\mathcal{C}$, the weight $w_\gamma$ satisfies
    \begin{equation}
        \abs{w_\gamma} \leq \left(\frac{1}{e^3\Delta\binom{r}{2}}\right)^\norm{\gamma}. \notag
    \end{equation}
\end{lemma}
\begin{proof}
    Fix a polymer $\gamma$. We have
    \begin{equation}
        \abs{w_\gamma} \leq \prod_{\varepsilon \in E(\gamma)}\norm{O_\varepsilon-\mathbb{I}} \leq \left(\frac{1}{e^3\Delta\binom{r}{2}}\right)^\norm{\gamma}, \notag
    \end{equation}
    completing the proof.
\end{proof}

\begin{lemma}
    \label{lemma:QuantumExpectationValueWeightAlgorithm}
    The weight $w_\gamma$ of a polymer $\gamma$ can be computed in time $\exp(O(\norm{\gamma}))$.
\end{lemma}
\begin{proof}
    The proof follows similarly to that of Lemma~\ref{lemma:QuantumProbabilityAmplitudeWeightAlgorithm}.
\end{proof}

Combining Theorem~\ref{theorem:ApproximationAlgorithmAbstractPolymerModelPartitionFunction} with Lemma~\ref{lemma:QuantumExpectationValueAbstractPolymerModel}, Lemma~\ref{lemma:QuantumExpectationValueWeightBound} and Lemma~\ref{lemma:QuantumExpectationValueWeightAlgorithm} proves Theorem~\ref{theorem:ApproximationAlgorithmQuantumExpectationValue}. We now show that the algorithmic condition of Theorem~\ref{theorem:ApproximationAlgorithmQuantumExpectationValue} is almost optimal in the sense of the zero freeness of the expectation value. This is achieved by a constructive argument based on an observation of Ref.~\cite{bravyi2021classical} and is formalised by the following theorem.
\begin{theorem}
    \label{theorem:ZeroExpectationValue}
    Fix $d\in\mathbb{Z}^+$ and $k\in\mathbb{Z}_{\geq2}$. There exists a hypergraph $G=(V, E)$, a quantum circuit $U_G$ of depth $d$ with each gate acting on at most $k$ qubits, and an operator $O$ satisfying $\norm{O_v-\mathbb{I}}\leq\frac{2}{k^d}$ for all $v \in V$, such that $\expval{O}_{U_G}=0$.
\end{theorem}
\begin{proof}
    Let $\ket{\psi_n}$ denote the state $\ket{\psi_n}\coloneqq\frac{1}{\sqrt{2}}\left(\ket{0^n}+\ket{1^n}\right)$. Note that there is a hypergraph $G$ and a quantum circuit $U_G$ of depth $d$ with each gate acting on at most $k$ qubits such that $\ket{\psi_{k^d}}=U_G\ket{0^\abs{G}}$. We consider the operator $O$ with $O_v=\mathbb{I}+i\tan\left(\frac{\pi}{2k^d}\right)Z_v$ for all $v \in V$. Then, we have
    \begin{align}
        \expval{O}_{U_G} &= \matrixel{0^\abs{G}}{{U_G}^\dag O_G U_G}{0^\abs{G}} \notag \\
        &= \matrixel{\psi_{k^d}}{\left[\prod_{v \in V}\left(\mathbb{I}+i\tan\left(\frac{\pi}{2k^d}\right)Z_v\right)\right]}{\psi_{k^d}} \notag \\ 
        &= \matrixel{\psi_{k^d}}{\left(\sum_{S \subseteq V}\prod_{v \in S}i\tan\left(\frac{\pi}{2k^d}\right)Z_v\right)}{\psi_{k^d}} \notag \\
        &= \frac{1}{2}\sum_{S \subseteq V}\left[i\tan\left(\frac{\pi}{2k^d}\right)\right]^\abs{S}\left[1+(-1)^\abs{S}\right] \notag \\
        &= \frac{1}{2}\left[\left(1+i\tan\left(\frac{\pi}{2k^d}\right)\right)^{k^d}+\left(1-i\tan\left(\frac{\pi}{2k^d}\right)\right)^{k^d}\right] \notag \\
        &= 0. \notag
    \end{align}
    Further, for all $v \in V$, we have
    \begin{equation}
        \norm{O_v-\mathbb{I}} = \abs{\tan\left(\frac{\pi}{2k^d}\right)} \leq \frac{2}{k^d}. \notag
    \end{equation}
    This completes the proof.
\end{proof}

\begin{remark}
    The operator in the proof of Theorem~\ref{theorem:ZeroExpectationValue} is not self-adjoint.
\end{remark}

\subsection{Partition Functions}
\label{section:PartitionFunctions}

In this section we study the problem of approximating partition functions of quantum spin systems. This problem is known to be \#P-hard in general~\cite{goldberg2017complexity}; however, we show that, for a class of quantum spin systems at high temperature, this problem admits an efficient approximation algorithm. Efficient approximation algorithms have previously been established for approximating partition functions of quantum spin systems at high temperature~\cite{harrow2020classical, kuwahara2020clustering, mann2021efficient} and for restricted classes at low temperature~\cite{helmuth2023efficient}. Our analysis closely follows that of Ref.~\cite{mann2021efficient} and can be viewed as a straightforward generalisation from the setting of bounded-degree graphs to bounded-degree bounded-rank multihypergraphs. This offers a simpler and slightly sharper analysis than Refs.~\cite{harrow2020classical, kuwahara2020clustering}. Further, we show that this algorithmic condition is optimal under complexity-theoretic assumptions.

We model a quantum spin system by a multihypergraph $G=(V, E)$. At each vertex $v$ of $G$ there is a $d$-dimensional Hilbert space $\mathcal{H}_v$ with $d<\infty$. The Hilbert space on the multihypergraph is given by $\mathcal{H}_G\coloneqq\bigotimes_{v \in V}\mathcal{H}_v$. An interaction $\Phi$ assigns a self-adjoint operator $\Phi(\varepsilon)$ on $\mathcal{H}_\varepsilon\coloneqq\bigotimes_{v \in \varepsilon}\mathcal{H}_v$ to each edge $\varepsilon$ of $G$. The Hamiltonian on $G$ is defined by $H_G\coloneqq\sum_{\varepsilon \in E}\Phi(\varepsilon)$. We are interested in the partition function $Z_G(\beta)$ at inverse temperature $\beta$, defined by $Z_G(\beta)\coloneqq\Tr\left[e^{-\beta H_G}\right]$. We shall assume that the trace is normalised so that $\Tr(\mathbb{I})=1$, which is equivalent to a rescaling the partition function by a multiplicative factor. Further, we shall assume that $\norm{\Phi(\varepsilon)}\leq1$ for all $\varepsilon \in E$, which is always possible by a rescaling of $\beta$. Our algorithmic result concerning the approximation of $Z_G(\beta)$ is as follows.
\begin{theorem}
    \label{theorem:ApproximationAlgorithmQuantumPartitionFunction}
    Fix $\Delta,r\in\mathbb{Z}_{\geq2}$. Let $G=(V, E)$ be a multihypergraph of maximum degree at most $\Delta$ and rank at most $r$, and let $\beta$ be a complex number such that
    \begin{equation}
        \abs{\beta} \leq \frac{1}{e^4\Delta\binom{r}{2}}. \notag
    \end{equation}
    Then the cluster expansion for $\log(Z_G(\beta))$ converges absolutely, $Z_G(\beta)\neq0$, and there is a fully polynomial-time approximation scheme for $Z_G(\beta)$.
\end{theorem}

\begin{remark}
    Theorem~\ref{theorem:ApproximationAlgorithmQuantumPartitionFunction} applies when $\beta$ is complex, which includes the case of time evolution.
\end{remark}

This offers a modest improvement over Ref.~\cite{harrow2020classical}, which established a quasi-polynomial time algorithm when $\abs{\beta}\leq\frac{1}{10e^2\Delta\binom{r}{2}}$ and over Ref.~\cite{kuwahara2020clustering}, which established a polynomial-time algorithm when $\abs{\beta}\leq\frac{1}{16e^4\Delta\binom{r}{2}}$. In the case when $G$ is a bounded-degree graph, we recover the results of Ref.~\cite{mann2021efficient}.

We prove Theorem~\ref{theorem:ApproximationAlgorithmQuantumPartitionFunction} by showing that the conditions required to apply Theorem~\ref{theorem:ApproximationAlgorithmAbstractPolymerModelPartitionFunction} are satisfied. That is, we show that (1) the partition function $Z_G(\beta)$ admits a suitable abstract polymer model representation, (2) the polymer weights satisfy the desired bound, and (3) the polymer weights can be computed in the desired time. This is achieved in the following three lemmas.
\begin{lemma}
    \label{lemma:QuantumPartitionFunctionAbstractPolymerModel}
    The partition function $Z_G(\beta)$ admits the following abstract polymer model representation.
    \begin{equation}
        Z_G(\beta) = \sum_{\Gamma\in\mathcal{G}}\prod_{\gamma\in\Gamma}w_\gamma, \notag
    \end{equation}
    where
    \begin{equation}
        w_\gamma \coloneqq (-1)^\norm{\gamma}\sum_{T \subseteq E(\gamma)}(-1)^\abs{T}\Tr\left[e^{-\beta\sum_{\varepsilon \in T}\Phi(\varepsilon)}\right]. \notag
    \end{equation}
\end{lemma}
\begin{proof}
    By applying Lemma~\ref{lemma:InclusionExclusion} with $f(E)=\Tr\left[e^{-\beta\sum_{\varepsilon \in E}\Phi(\varepsilon)}\right]$, we have
    \begin{align}
        Z_G(\beta) &= \Tr\left[e^{-\beta H_G}\right] \notag \\
        &= \sum_{S \subseteq E}(-1)^\abs{S}\sum_{T \subseteq S}(-1)^\abs{T}\Tr\left[e^{-\beta\sum_{\varepsilon \in T}\Phi(\varepsilon)}\right]. \notag
    \end{align}
    For a subset $S \subseteq E$, let $\Gamma_S$ denote the maximally connected components of $S$. By factorising over these components, we have
    \begin{align}
        Z_G(\beta) &= \sum_{S \subseteq E}\prod_{\gamma \in \Gamma_S}(-1)^\norm{\gamma}\sum_{T \subseteq E(\gamma)}(-1)^\abs{T}\Tr\left[e^{-\beta\sum_{\varepsilon \in T}\Phi(\varepsilon)}\right] \notag \\
        &= \sum_{S \subseteq E}\prod_{\gamma \in \Gamma_S}w_\gamma \notag \\
        &= \sum_{\Gamma\in\mathcal{G}}\prod_{\gamma\in\Gamma}w_\gamma. \notag
    \end{align}
    This completes the proof.
\end{proof}

\begin{lemma}
    \label{lemma:QuantumPartitionFunctionWeightBound}
    Fix $\Delta,r\in\mathbb{Z}_{\geq2}$. Let $G=(V, E)$ be a multihypergraph of maximum degree at most $\Delta$ and rank at most $r$, and let $\beta$ be a complex number such that
    \begin{equation}
        \abs{\beta} \leq \frac{1}{e^4\Delta\binom{r}{2}}. \notag
    \end{equation}
    Then, for all polymers $\gamma\in\mathcal{C}$, the weight $w_\gamma$ satisfies
    \begin{equation}
        \abs{w_\gamma} \leq \left(\frac{1}{e^3\Delta\binom{r}{2}}\right)^\norm{\gamma}. \notag
    \end{equation}
\end{lemma}
\begin{proof}
    Fix a polymer $\gamma$. Let $P$ denote the set of all sequences of edges in $\gamma$. By the Taylor series,
    \begin{equation}
        \abs{w_\gamma} \leq \norm{\sum_{T \subseteq E(\gamma)}(-1)^\abs{T}e^{-\beta\sum_{\varepsilon \in T}\Phi(\varepsilon)}} \leq \sum_{\substack{\rho \in P \\ \text{supp}(\rho)=\gamma}}\frac{\abs{\beta}^\abs{\rho}}{\abs{\rho}!}\prod_{\varepsilon \in \rho}\norm{\Phi(\varepsilon)} \leq \sum_{\substack{\rho \in P \\ \text{supp}(\rho)=\gamma}}\frac{\abs{\beta}^\abs{\rho}}{\abs{\rho}!}. \notag
    \end{equation}
    There are precisely $\genfrac{\{}{\}}{0pt}{}{n}{\norm{\gamma}}\norm{\gamma}!$ sequences $\rho$ of length $n$ whose support is $\gamma$, where $\genfrac{\{}{\}}{0pt}{}{n}{\norm{\gamma}}$ denotes the Stirling number of the second kind. Hence, we may write
    \begin{equation}
        \abs{w_\gamma} \leq \sum_{n=\norm{\gamma}}^\infty\genfrac{\{}{\}}{0pt}{}{n}{\norm{\gamma}}\frac{\norm{\gamma}!}{n!}\abs{\beta}^n = (e^\abs{\beta}-1)^\norm{\gamma}, \notag
    \end{equation}
    where we have used the identity $\sum_{n=k}^\infty\genfrac{\{}{\}}{0pt}{}{n}{k}\frac{x^n}{n!}=\frac{(e^x-1)^k}{k^!}$. By taking $\abs{\beta}\leq\left(e^4\Delta\binom{r}{2}\right)^{-1}$, we have
    \begin{equation}
        \abs{w_\gamma} \leq \left(\frac{1}{e^3\Delta\binom{r}{2}}\right)^\norm{\gamma}, \notag
    \end{equation}
    completing the proof.
\end{proof}

\begin{lemma}
    \label{lemma:QuantumPartitionFunctionWeightAlgorithm}
    The weight $w_\gamma$ of a polymer $\gamma$ can be computed in time $\exp(O(\norm{\gamma}))$.
\end{lemma}
\begin{proof}
    The sum is over all subsets $T$ of $E(\gamma)$, of which there are $2^\norm{\gamma}$. For each of these subsets $T$, the trace may be evaluated in time $\exp(O(\norm{\gamma}))$ by diagonalising the sum of interactions.
\end{proof}

Combining Theorem~\ref{theorem:ApproximationAlgorithmAbstractPolymerModelPartitionFunction} with Lemma~\ref{lemma:QuantumPartitionFunctionAbstractPolymerModel}, Lemma~\ref{lemma:QuantumPartitionFunctionWeightBound}, and Lemma~\ref{lemma:QuantumPartitionFunctionWeightAlgorithm} proves Theorem~\ref{theorem:ApproximationAlgorithmQuantumPartitionFunction}. We now show that the algorithmic condition of Theorem~\ref{theorem:ApproximationAlgorithmQuantumPartitionFunction} is optimal in the case of multigraphs under complexity-theoretic assumptions. This is achieved by establishing a hardness of approximation result for the partition function $Z_G(\beta)$ at imaginary temperature, i.e., $\beta=i\theta$ for $\theta\in\mathbb{R}$. Our hardness result concerning the approximation of $Z_G(i\theta)$ is as follows.
\begin{theorem}
    \label{theorem:ApproximationHardnessQuantumPartitionFunction}
    Fix $\epsilon>0$, $\Delta\in\mathbb{Z}_{\geq3}$, and $\theta\in\mathbb{R}$ such that $\abs{\theta}\geq\frac{3\pi}{5(\Delta-2)}$. It is \emph{\#P-hard} to approximate the partition function $Z_G(i\theta)$ up to a multiplicative $\epsilon$-approximation on multigraphs of maximum degree at most $\Delta$.
\end{theorem}
\begin{proof}
    Our proof is based on a reduction from the Ising model partition function. We consider quantum spin systems on multigraphs with a two-dimensional Hilbert space at each vertex and self-adjoint operators of the form $\Phi(\varepsilon)=\phi(\varepsilon)\bigotimes_{v \in \varepsilon}Z_v$, where $\phi(\varepsilon)$ is a real number satisfying $\abs{\phi(\varepsilon)}\leq1$. We have
    \begin{align}
        Z_G(i\theta) &= \Tr\left[\prod_{\varepsilon \in E}e^{-i\theta\phi(\varepsilon)\bigotimes_{v \in \varepsilon}Z_v}\right] \notag \\
        &= \frac{1}{2^\abs{G}}\sum_{\sigma\in\{-1,+1\}^V}\prod_{\{u,v\} \in E}e^{-i\theta\phi\left(\{u,v\}\right)\sigma_u\sigma_v} \notag \\
        &= \frac{1}{2^\abs{G}}Z_{\mathrm{Ising}}(G;i\theta). \notag
    \end{align}
    The proof then follows from Theorem~\ref{theorem:IsingModelPartitionFunctionApproximationHardness}.
\end{proof}

We note that a hardness of approximation result with similar bounds may be obtained for real temperature under the assumption that RP is not equal to NP via the results of Refs.~\cite{sly2012computational, galanis2016inapproximability}. Our results establish a computational complexity transition from P to \#P-hard for the problem of approximating partition functions.

\subsection{Thermal Expectation Values}
\label{section:ThermalExpectationValues}

In this section we study the problem of approximating thermal expectation values of quantum spin systems. This problem is known to be \#P-hard in general~\cite{bravyi2022quantum}; however, we show that, for a class of quantum spin systems at high temperature with positive-semidefinite operators, this problem admits an efficient approximation algorithm. This setting was studied in Ref.~\cite{wild2023classical}, which established an efficient approximation algorithm for this problem. Our approach offers a similar but slightly sharper analysis. Further, we show that this algorithmic condition is optimal in the sense of zero freeness.

We model a quantum spin system by a multihypergraph $G=(V, E)$ as in Section~\ref{section:PartitionFunctions}. An operator $\Psi$ assigns a positive-semidefinite operator $\Psi(v)$ on $\mathcal{H}_v$ to each vertex $v$ of $G$. The operator $\Psi_G$ on $G$ is defined by $\Psi_G\coloneqq\prod_{v \in V}\Psi(v)$. We are interested in the thermal expectation value $\expval{\Psi}_G(\beta)$ at inverse temperature $\beta$, defined by $\expval{\Psi}_G(\beta)\coloneqq\frac{Z_G^\Psi(\beta)}{Z_G(\beta)}$, where $Z_G^\Psi(\beta)\coloneqq\Tr\left[\Psi_Ge^{-\beta H_G}\right]$. We shall assume that the positive-semidefinite operators are normalised so that $\Tr(\Psi_v)=1$ for all $v \in V$, which is equivalent to a rescaling of the thermal expectation value by a multiplicative factor. Our algorithmic result concerning the approximation of $\expval{\Psi}_G(\beta)$ is as follows.
\begin{theorem}
    \label{theorem:ApproximationAlgorithmThermalExpectationValues}
    Fix $\Delta,r\in\mathbb{Z}_{\geq2}$. Let $G=(V, E)$ be a multihypergraph of maximum degree at most $\Delta$ and rank at most $r$, and let $\beta$ be a complex number such that
    \begin{equation}
        \abs{\beta} \leq \frac{1}{e^4\Delta\binom{r}{2}}. \notag
    \end{equation}
    Then the cluster expansion for $\log(\expval{\Psi}_G(\beta))$ converges absolutely, $\expval{\Psi}_G(\beta)\neq0$, and there is a fully polynomial-time approximation scheme for $\expval{\Psi}_G(\beta)$.
\end{theorem}

\begin{remark}
    Theorem~\ref{theorem:ApproximationAlgorithmThermalExpectationValues} applies when $\beta$ is complex, which includes the case of time evolution.
\end{remark}

This offers a modest improvement over Ref.~\cite{wild2023classical} when $\Delta\geq4$, which established a polynomial-time algorithm when $\abs{\beta}\leq\frac{1}{2e^2(\Delta-1)r(\Delta r-r+1)}$. By using the slightly sharper bound given in the remark subsequent to Lemma~\ref{lemma:ConnectedSubraphCount}, we obtain an improvement when $\Delta\geq3$. We note that efficient approximations algorithms may be established when the observable appears in the Hamiltonian under different assumptions.

We prove Theorem~\ref{theorem:ApproximationAlgorithmThermalExpectationValues} by showing that the conditions required to apply Theorem~\ref{theorem:ApproximationAlgorithmAbstractPolymerModelPartitionFunction} are satisfied and then combining this with Theorem~\ref{theorem:ApproximationAlgorithmQuantumPartitionFunction}. That is, we show that (1) $Z_G^\Psi(\beta)$ admits a suitable abstract polymer model representation, (2) the polymer weights satisfy the desired bound, and (3) the polymer weights can be computed in the desired time. This is achieved in the following three lemmas.
\begin{lemma}
    \label{lemma:ThermalExpectationValuesAbstractPolymerModel}
    $Z_G^\Psi(\beta)$ admits the following abstract polymer model representation.
    \begin{equation}
        Z_G^\Psi(\beta) = \sum_{\Gamma\in\mathcal{G}}\prod_{\gamma\in\Gamma}w_\gamma, \notag
    \end{equation}
    where
    \begin{equation}
        w_\gamma \coloneqq (-1)^\norm{\gamma}\sum_{T \subseteq E(\gamma)}(-1)^\abs{T}\Tr\left[\Psi_\gamma e^{-\beta\sum_{\varepsilon \in T}\Phi(\varepsilon)}\right]. \notag
    \end{equation}
\end{lemma}
\begin{proof}
    The proof follows similarly to that of Lemma~\ref{lemma:QuantumPartitionFunctionAbstractPolymerModel}.
\end{proof}

\begin{lemma}
    \label{lemma:ThermalExpectationValuesWeightBound}
    Fix $\Delta,r\in\mathbb{Z}_{\geq2}$. Let $G=(V, E)$ be a multihypergraph of maximum degree at most $\Delta$ and rank at most $r$, and let $\beta$ be a complex number such that
    \begin{equation}
        \abs{\beta} \leq \frac{1}{e^4\Delta\binom{r}{2}}. \notag
    \end{equation}
    Then, for all polymers $\gamma\in\mathcal{C}$, the weight $w_\gamma$ satisfies
    \begin{equation}
        \abs{w_\gamma} \leq \left(\frac{1}{e^3\Delta\binom{r}{2}}\right)^\norm{\gamma}. \notag
    \end{equation}
\end{lemma}
\begin{proof}
    The proof follows similarly to that of Lemma~\ref{lemma:QuantumPartitionFunctionWeightBound}.
\end{proof}

\begin{lemma}
    \label{lemma:ThermalExpectationValuesWeightAlgorithm}
    The weight $w_\gamma$ of a polymer $\gamma$ can be computed in time $\exp(O(\norm{\gamma}))$.
\end{lemma}
\begin{proof}
    The sum is over all subsets $T$ of $E(\gamma)$, of which there are $2^\norm{\gamma}$. For each of these subsets $T$, the trace may be evaluated in time $\exp(O(\norm{\gamma}))$ by diagonalising the sum of interactions and matrix multiplication.
\end{proof}

Combining Theorem~\ref{theorem:ApproximationAlgorithmAbstractPolymerModelPartitionFunction} with Lemma~\ref{lemma:ThermalExpectationValuesAbstractPolymerModel}, Lemma~\ref{lemma:ThermalExpectationValuesWeightBound}, Lemma~\ref{lemma:ThermalExpectationValuesWeightAlgorithm}, and Theorem~\ref{theorem:ApproximationAlgorithmQuantumPartitionFunction} proves Theorem~\ref{theorem:ApproximationAlgorithmThermalExpectationValues}. We now show that the algorithmic condition of Theorem~\ref{theorem:ApproximationAlgorithmThermalExpectationValues} is optimal in the case of multigraphs in the sense of the zero freeness of the thermal expectation value. This is achieved by a straightforward constructive argument and is formalised by the following theorem.
\begin{theorem}
    \label{theorem:ZeroThermalExpectationValue}
    Fix $\Delta\in\mathbb{Z}^+$. There exists a multigraph $G=(V, E)$ of maximum degree $\Delta$, an operator $O$, and a self-adjoint operator $\Phi$, such that $\expval{\Psi}_G(\beta)=0$ with $\beta=\frac{i\pi}{\Delta}$.
\end{theorem}
\begin{proof}
    We consider a quantum spin system on a multigraph comprising a single multiedge with a two-dimensional Hilbert space at each vertex. Further, we consider the operator $\Psi$ with $\Psi(v)=\ketbra{0}{0}_v$ for all $v \in V$ and the self-adjoint operator $\Phi$ with $\Phi(\varepsilon)=\frac{1}{4}\left(\bigotimes_{v \in \varepsilon}X_v-\bigotimes_{v \in \varepsilon}Y_v-\bigotimes_{v \in \varepsilon}Z_v\right)$ for all $\varepsilon \in E$. Then, we have
    \begin{align}
        \expval{\Psi}_G(\beta) &= \frac{\Tr\left[\Psi_Ge^{-\beta H_G}\right]}{\Tr\left[e^{-\beta H_G}\right]} \notag \\
        &= \frac{\matrixel{00}{e^{-\frac{i\pi}{4}\left(X{\otimes}X-Y{\otimes}Y-Z{\otimes}Z\right)}}{00}}{\Tr\left[e^{-\frac{i\pi}{4}\left(X{\otimes}X-Y{\otimes}Y-Z{\otimes}Z\right)}\right]} \notag \\ 
        &= 0. \notag
    \end{align}
    This completes the proof.
\end{proof}

\section{Conclusion \& Outlook}
\label{section:ConclusionAndOutlook}

We have established a general framework for developing approximation algorithms and hardness of approximation results for a class of counting problems. We applied this framework to obtain efficient approximation algorithms and hardness of approximation results for several quantum problems under certain algorithmic conditions. 

In particular, we obtained efficient approximation algorithms for (1) approximating probability amplitudes of a class of quantum circuits close to the identity, (2) approximating expectation values of a class of quantum circuits with operators close to the identity, (3) approximating partition functions of a class of quantum spin systems at high temperature, and (4) approximating thermal expectation values of a class of quantum spin systems at high temperature with positive-semidefinite operators. Further, we obtained hardness of approximation results for approximating probability amplitudes of quantum circuits and partition functions of quantum spin systems.

Our results established a computational complexity transition for the problems of approximating probability amplitudes of quantum circuits and partition functions of quantum spin systems and showed that our algorithmic conditions for these problems are optimal under complexity-theoretic assumptions. Finally, we showed that our algorithmic condition is almost optimal for expectation values and optimal for thermal expectation values in the sense of zero freeness.

It would be interesting to identify other quantum problems to which our framework applies and to explore whether our results can be extended using ideas from duality~\cite{apel2022mathematical}. Further, it is an intriguing open problem to not only identify the exact points of a computational complexity transition for these problems, as is known for the Ising model at real temperature~\cite{sly2012computational, sinclair2014approximation, galanis2016inapproximability}, but also to develop methods to experimentally probe these transitions. Finally, it would be interesting to obtain algorithms with an improved runtime, for example, via the Markov chain polymer approach of Ref.~\cite{chen2019fast}.

\section*{Acknowledgements}

We thank Tyler Helmuth for helpful discussions. RLM was supported by the QuantERA ERA-NET Cofund in Quantum Technologies implemented within the European Union's Horizon 2020 Programme (QuantAlgo project), EPSRC grants EP/L021005/1, EP/R043957/1, and EP/T001062/1, and the ARC Centre of Excellence for Quantum Computation and Communication Technology (CQC2T), project number CE170100012. RMM was supported by the Additional Funding Programme for Mathematical Sciences, delivered by EPSRC (EP/V521917/1) and the Heilbronn Institute for Mathematical Research. No new data were created during this study.

\bibliography{bibliography}

\begin{thebibliography}{31}%
\makeatletter
\providecommand \@ifxundefined [1]{%
 \@ifx{#1\undefined}
}%
\providecommand \@ifnum [1]{%
 \ifnum #1\expandafter \@firstoftwo
 \else \expandafter \@secondoftwo
 \fi
}%
\providecommand \@ifx [1]{%
 \ifx #1\expandafter \@firstoftwo
 \else \expandafter \@secondoftwo
 \fi
}%
\providecommand \natexlab [1]{#1}%
\providecommand \enquote  [1]{``#1''}%
\providecommand \bibnamefont  [1]{#1}%
\providecommand \bibfnamefont [1]{#1}%
\providecommand \citenamefont [1]{#1}%
\providecommand \href@noop [0]{\@secondoftwo}%
\providecommand \href [0]{\begingroup \@sanitize@url \@href}%
\providecommand \@href[1]{\@@startlink{#1}\@@href}%
\providecommand \@@href[1]{\endgroup#1\@@endlink}%
\providecommand \@sanitize@url [0]{\catcode `\\12\catcode `\$12\catcode `\&12\catcode `\#12\catcode `\^12\catcode `\_12\catcode `\%12\relax}%
\providecommand \@@startlink[1]{}%
\providecommand \@@endlink[0]{}%
\providecommand \url  [0]{\begingroup\@sanitize@url \@url }%
\providecommand \@url [1]{\endgroup\@href {#1}{\urlprefix }}%
\providecommand \urlprefix  [0]{URL }%
\providecommand \Eprint [0]{\href }%
\providecommand \doibase [0]{https://doi.org/}%
\providecommand \selectlanguage [0]{\@gobble}%
\providecommand \bibinfo  [0]{\@secondoftwo}%
\providecommand \bibfield  [0]{\@secondoftwo}%
\providecommand \translation [1]{[#1]}%
\providecommand \BibitemOpen [0]{}%
\providecommand \bibitemStop [0]{}%
\providecommand \bibitemNoStop [0]{.\EOS\space}%
\providecommand \EOS [0]{\spacefactor3000\relax}%
\providecommand \BibitemShut  [1]{\csname bibitem#1\endcsname}%
\let\auto@bib@innerbib\@empty
\bibitem [{\citenamefont {Bravyi}\ \emph {et~al.}(2021)\citenamefont {Bravyi}, \citenamefont {Gosset},\ and\ \citenamefont {Movassagh}}]{bravyi2021classical}%
  \BibitemOpen
  \bibfield  {author} {\bibinfo {author} {\bibfnamefont {S.}~\bibnamefont {Bravyi}}, \bibinfo {author} {\bibfnamefont {D.}~\bibnamefont {Gosset}},\ and\ \bibinfo {author} {\bibfnamefont {R.}~\bibnamefont {Movassagh}},\ }\href {https://doi.org/10.1038/s41567-020-01109-8} {\bibfield  {journal} {\bibinfo  {journal} {Nature Physics}\ }\textbf {\bibinfo {volume} {17}},\ \bibinfo {pages} {337} (\bibinfo {year} {2021})},\ \Eprint {https://arxiv.org/abs/1909.11485} {arXiv:1909.11485} \BibitemShut {NoStop}%
\bibitem [{\citenamefont {Kuwahara}\ \emph {et~al.}(2020)\citenamefont {Kuwahara}, \citenamefont {Kato},\ and\ \citenamefont {Brand{\~a}o}}]{kuwahara2020clustering}%
  \BibitemOpen
  \bibfield  {author} {\bibinfo {author} {\bibfnamefont {T.}~\bibnamefont {Kuwahara}}, \bibinfo {author} {\bibfnamefont {K.}~\bibnamefont {Kato}},\ and\ \bibinfo {author} {\bibfnamefont {F.~G.}\ \bibnamefont {Brand{\~a}o}},\ }\href {https://doi.org/10.1103/physrevlett.124.220601} {\bibfield  {journal} {\bibinfo  {journal} {Physical Review Letters}\ }\textbf {\bibinfo {volume} {124}},\ \bibinfo {pages} {220601} (\bibinfo {year} {2020})},\ \Eprint {https://arxiv.org/abs/1910.09425} {arXiv:1910.09425} \BibitemShut {NoStop}%
\bibitem [{\citenamefont {Wild}\ and\ \citenamefont {Alhambra}(2023)}]{wild2023classical}%
  \BibitemOpen
  \bibfield  {author} {\bibinfo {author} {\bibfnamefont {D.~S.}\ \bibnamefont {Wild}}\ and\ \bibinfo {author} {\bibfnamefont {{\'A}.~M.}\ \bibnamefont {Alhambra}},\ }\href {https://doi.org/10.1103/PRXQuantum.4.020340} {\bibfield  {journal} {\bibinfo  {journal} {PRX Quantum}\ }\textbf {\bibinfo {volume} {4}},\ \bibinfo {pages} {020340} (\bibinfo {year} {2023})},\ \Eprint {https://arxiv.org/abs/2210.11490} {arXiv:2210.11490} \BibitemShut {NoStop}%
\bibitem [{\citenamefont {Koteck{\'y}}\ and\ \citenamefont {Preiss}(1986)}]{kotecky1986cluster}%
  \BibitemOpen
  \bibfield  {author} {\bibinfo {author} {\bibfnamefont {R.}~\bibnamefont {Koteck{\'y}}}\ and\ \bibinfo {author} {\bibfnamefont {D.}~\bibnamefont {Preiss}},\ }\href {https://doi.org/10.1007/bf01211762} {\bibfield  {journal} {\bibinfo  {journal} {Communications in Mathematical Physics}\ }\textbf {\bibinfo {volume} {103}},\ \bibinfo {pages} {491} (\bibinfo {year} {1986})}\BibitemShut {NoStop}%
\bibitem [{\citenamefont {Helmuth}\ \emph {et~al.}(2020)\citenamefont {Helmuth}, \citenamefont {Perkins},\ and\ \citenamefont {Regts}}]{helmuth2020algorithmic}%
  \BibitemOpen
  \bibfield  {author} {\bibinfo {author} {\bibfnamefont {T.}~\bibnamefont {Helmuth}}, \bibinfo {author} {\bibfnamefont {W.}~\bibnamefont {Perkins}},\ and\ \bibinfo {author} {\bibfnamefont {G.}~\bibnamefont {Regts}},\ }\href {https://doi.org/10.1007/s00440-019-00928-y} {\bibfield  {journal} {\bibinfo  {journal} {Probability Theory and Related Fields}\ }\textbf {\bibinfo {volume} {176}},\ \bibinfo {pages} {851} (\bibinfo {year} {2020})},\ \Eprint {https://arxiv.org/abs/1806.11548} {arXiv:1806.11548} \BibitemShut {NoStop}%
\bibitem [{\citenamefont {Borgs}\ \emph {et~al.}(2020)\citenamefont {Borgs}, \citenamefont {Chayes}, \citenamefont {Helmuth}, \citenamefont {Perkins},\ and\ \citenamefont {Tetali}}]{borgs2020efficient}%
  \BibitemOpen
  \bibfield  {author} {\bibinfo {author} {\bibfnamefont {C.}~\bibnamefont {Borgs}}, \bibinfo {author} {\bibfnamefont {J.}~\bibnamefont {Chayes}}, \bibinfo {author} {\bibfnamefont {T.}~\bibnamefont {Helmuth}}, \bibinfo {author} {\bibfnamefont {W.}~\bibnamefont {Perkins}},\ and\ \bibinfo {author} {\bibfnamefont {P.}~\bibnamefont {Tetali}},\ }in\ \href {https://doi.org/10.1145/3357713.3384271} {\emph {\bibinfo {booktitle} {Proceedings of the 52nd Annual ACM SIGACT Symposium on Theory of Computing}}}\ (\bibinfo {organization} {ACM},\ \bibinfo {year} {2020})\ pp.\ \bibinfo {pages} {738--751},\ \Eprint {https://arxiv.org/abs/1909.09298} {arXiv:1909.09298} \BibitemShut {NoStop}%
\bibitem [{\citenamefont {Patel}\ and\ \citenamefont {Regts}(2017)}]{patel2017deterministic}%
  \BibitemOpen
  \bibfield  {author} {\bibinfo {author} {\bibfnamefont {V.}~\bibnamefont {Patel}}\ and\ \bibinfo {author} {\bibfnamefont {G.}~\bibnamefont {Regts}},\ }\href {https://doi.org/10.1137/16m1101003} {\bibfield  {journal} {\bibinfo  {journal} {SIAM Journal on Computing}\ }\textbf {\bibinfo {volume} {46}},\ \bibinfo {pages} {1893} (\bibinfo {year} {2017})},\ \Eprint {https://arxiv.org/abs/1607.01167} {arXiv:1607.01167} \BibitemShut {NoStop}%
\bibitem [{\citenamefont {Barvinok}(2016)}]{barvinok2016combinatorics}%
  \BibitemOpen
  \bibfield  {author} {\bibinfo {author} {\bibfnamefont {A.}~\bibnamefont {Barvinok}},\ }\href {https://doi.org/10.1007/978-3-319-51829-9} {\emph {\bibinfo {title} {Combinatorics and Complexity of Partition Functions}}},\ Vol.\ \bibinfo {volume} {274}\ (\bibinfo  {publisher} {Springer},\ \bibinfo {year} {2016})\BibitemShut {NoStop}%
\bibitem [{\citenamefont {Patel}\ and\ \citenamefont {Regts}(2022)}]{patel2022approximate}%
  \BibitemOpen
  \bibfield  {author} {\bibinfo {author} {\bibfnamefont {V.}~\bibnamefont {Patel}}\ and\ \bibinfo {author} {\bibfnamefont {G.}~\bibnamefont {Regts}},\ }\href@noop {} {\bibfield  {journal} {\bibinfo  {journal} {Bulletin of EATCS}\ }\textbf {\bibinfo {volume} {138}} (\bibinfo {year} {2022})},\ \Eprint {https://arxiv.org/abs/2212.08143} {arXiv:2212.08143} \BibitemShut {NoStop}%
\bibitem [{\citenamefont {Graham}\ \emph {et~al.}(1995)\citenamefont {Graham}, \citenamefont {Gr\"{o}tschel},\ and\ \citenamefont {Lov\'{a}sz}}]{graham1995handbook}%
  \BibitemOpen
  \bibfield  {author} {\bibinfo {author} {\bibfnamefont {R.~L.}\ \bibnamefont {Graham}}, \bibinfo {author} {\bibfnamefont {M.}~\bibnamefont {Gr\"{o}tschel}},\ and\ \bibinfo {author} {\bibfnamefont {L.}~\bibnamefont {Lov\'{a}sz}},\ }\href@noop {} {\emph {\bibinfo {title} {Handbook of Combinatorics}}},\ Vol.~\bibinfo {volume} {2}\ (\bibinfo  {publisher} {Elsevier},\ \bibinfo {year} {1995})\BibitemShut {NoStop}%
\bibitem [{\citenamefont {Friedli}\ and\ \citenamefont {Velenik}(2017)}]{friedli2017statistical}%
  \BibitemOpen
  \bibfield  {author} {\bibinfo {author} {\bibfnamefont {S.}~\bibnamefont {Friedli}}\ and\ \bibinfo {author} {\bibfnamefont {Y.}~\bibnamefont {Velenik}},\ }\href {https://doi.org/10.1017/9781316882603} {\emph {\bibinfo {title} {Statistical Mechanics of Lattice Systems: A Concrete Mathematical Introduction}}}\ (\bibinfo  {publisher} {Cambridge University Press},\ \bibinfo {year} {2017})\BibitemShut {NoStop}%
\bibitem [{\citenamefont {Fern{\'a}ndez}\ and\ \citenamefont {Procacci}(2007)}]{fernandez2007cluster}%
  \BibitemOpen
  \bibfield  {author} {\bibinfo {author} {\bibfnamefont {R.}~\bibnamefont {Fern{\'a}ndez}}\ and\ \bibinfo {author} {\bibfnamefont {A.}~\bibnamefont {Procacci}},\ }\href {https://doi.org/10.1007/s00220-007-0279-2} {\bibfield  {journal} {\bibinfo  {journal} {Communications in Mathematical Physics}\ }\textbf {\bibinfo {volume} {274}},\ \bibinfo {pages} {123} (\bibinfo {year} {2007})},\ \Eprint {https://arxiv.org/abs/math-ph/0605041} {arXiv:math-ph/0605041} \BibitemShut {NoStop}%
\bibitem [{\citenamefont {Shearer}(1985)}]{shearer1985problem}%
  \BibitemOpen
  \bibfield  {author} {\bibinfo {author} {\bibfnamefont {J.~B.}\ \bibnamefont {Shearer}},\ }\href {https://doi.org/10.1007/BF02579368} {\bibfield  {journal} {\bibinfo  {journal} {Combinatorica}\ }\textbf {\bibinfo {volume} {5}},\ \bibinfo {pages} {241} (\bibinfo {year} {1985})}\BibitemShut {NoStop}%
\bibitem [{\citenamefont {Scott}\ and\ \citenamefont {Sokal}(2005)}]{scott2005repulsive}%
  \BibitemOpen
  \bibfield  {author} {\bibinfo {author} {\bibfnamefont {A.~D.}\ \bibnamefont {Scott}}\ and\ \bibinfo {author} {\bibfnamefont {A.~D.}\ \bibnamefont {Sokal}},\ }\href {https://doi.org/10.1007/s10955-004-2055-4} {\bibfield  {journal} {\bibinfo  {journal} {Journal of Statistical Physics}\ }\textbf {\bibinfo {volume} {118}},\ \bibinfo {pages} {1151} (\bibinfo {year} {2005})},\ \Eprint {https://arxiv.org/abs/cond-mat/0309352} {arXiv:cond-mat/0309352} \BibitemShut {NoStop}%
\bibitem [{\citenamefont {Galvin}\ \emph {et~al.}(2022)\citenamefont {Galvin}, \citenamefont {McKinley}, \citenamefont {Perkins}, \citenamefont {Sarantis},\ and\ \citenamefont {Tetali}}]{galvin2022zeroes}%
  \BibitemOpen
  \bibfield  {author} {\bibinfo {author} {\bibfnamefont {D.}~\bibnamefont {Galvin}}, \bibinfo {author} {\bibfnamefont {G.}~\bibnamefont {McKinley}}, \bibinfo {author} {\bibfnamefont {W.}~\bibnamefont {Perkins}}, \bibinfo {author} {\bibfnamefont {M.}~\bibnamefont {Sarantis}},\ and\ \bibinfo {author} {\bibfnamefont {P.}~\bibnamefont {Tetali}},\ }\href@noop {} {\bibfield  {journal} {\bibinfo  {journal} {arXiv e-prints}\ } (\bibinfo {year} {2022})},\ \Eprint {https://arxiv.org/abs/2211.00464} {arXiv:2211.00464} \BibitemShut {NoStop}%
\bibitem [{\citenamefont {Bencs}\ and\ \citenamefont {Buys}(2023)}]{bencs2023optimal}%
  \BibitemOpen
  \bibfield  {author} {\bibinfo {author} {\bibfnamefont {F.}~\bibnamefont {Bencs}}\ and\ \bibinfo {author} {\bibfnamefont {P.}~\bibnamefont {Buys}},\ }\href@noop {} {\bibfield  {journal} {\bibinfo  {journal} {arXiv e-prints}\ } (\bibinfo {year} {2023})},\ \Eprint {https://arxiv.org/abs/2306.00122} {arXiv:2306.00122} \BibitemShut {NoStop}%
\bibitem [{\citenamefont {Arora}\ and\ \citenamefont {Barak}(2009)}]{arora2009computational}%
  \BibitemOpen
  \bibfield  {author} {\bibinfo {author} {\bibfnamefont {S.}~\bibnamefont {Arora}}\ and\ \bibinfo {author} {\bibfnamefont {B.}~\bibnamefont {Barak}},\ }\href {https://doi.org/10.1017/CBO9780511804090} {\emph {\bibinfo {title} {Computational Complexity: A Modern Approach}}}\ (\bibinfo  {publisher} {Cambridge University Press},\ \bibinfo {year} {2009})\BibitemShut {NoStop}%
\bibitem [{\citenamefont {Borgs}\ \emph {et~al.}(2013)\citenamefont {Borgs}, \citenamefont {Chayes}, \citenamefont {Kahn},\ and\ \citenamefont {Lov{\'a}sz}}]{borgs2013left}%
  \BibitemOpen
  \bibfield  {author} {\bibinfo {author} {\bibfnamefont {C.}~\bibnamefont {Borgs}}, \bibinfo {author} {\bibfnamefont {J.}~\bibnamefont {Chayes}}, \bibinfo {author} {\bibfnamefont {J.}~\bibnamefont {Kahn}},\ and\ \bibinfo {author} {\bibfnamefont {L.}~\bibnamefont {Lov{\'a}sz}},\ }\href {https://doi.org/10.1002/rsa.20414} {\bibfield  {journal} {\bibinfo  {journal} {Random Structures \& Algorithms}\ }\textbf {\bibinfo {volume} {42}},\ \bibinfo {pages} {1} (\bibinfo {year} {2013})},\ \Eprint {https://arxiv.org/abs/1002.0115} {arXiv:1002.0115} \BibitemShut {NoStop}%
\bibitem [{\citenamefont {Bj{\"o}rklund}\ \emph {et~al.}(2008)\citenamefont {Bj{\"o}rklund}, \citenamefont {Husfeldt}, \citenamefont {Kaski},\ and\ \citenamefont {Koivisto}}]{bjorklund2008computing}%
  \BibitemOpen
  \bibfield  {author} {\bibinfo {author} {\bibfnamefont {A.}~\bibnamefont {Bj{\"o}rklund}}, \bibinfo {author} {\bibfnamefont {T.}~\bibnamefont {Husfeldt}}, \bibinfo {author} {\bibfnamefont {P.}~\bibnamefont {Kaski}},\ and\ \bibinfo {author} {\bibfnamefont {M.}~\bibnamefont {Koivisto}},\ }in\ \href {https://doi.org/10.1109/FOCS.2008.40} {\emph {\bibinfo {booktitle} {49th Annual IEEE Symposium on Foundations of Computer Science}}}\ (\bibinfo {organization} {IEEE},\ \bibinfo {year} {2008})\ pp.\ \bibinfo {pages} {677--686},\ \Eprint {https://arxiv.org/abs/0711.2585} {arXiv:0711.2585} \BibitemShut {NoStop}%
\bibitem [{\citenamefont {Galanis}\ \emph {et~al.}(2022)\citenamefont {Galanis}, \citenamefont {Goldberg},\ and\ \citenamefont {Herrera-Poyatos}}]{galanis2022complexity}%
  \BibitemOpen
  \bibfield  {author} {\bibinfo {author} {\bibfnamefont {A.}~\bibnamefont {Galanis}}, \bibinfo {author} {\bibfnamefont {L.~A.}\ \bibnamefont {Goldberg}},\ and\ \bibinfo {author} {\bibfnamefont {A.}~\bibnamefont {Herrera-Poyatos}},\ }\href {https://doi.org/10.1137/21M1454043} {\bibfield  {journal} {\bibinfo  {journal} {SIAM Journal on Discrete Mathematics}\ }\textbf {\bibinfo {volume} {36}},\ \bibinfo {pages} {2159} (\bibinfo {year} {2022})},\ \Eprint {https://arxiv.org/abs/2105.00287} {arXiv:2105.00287} \BibitemShut {NoStop}%
\bibitem [{\citenamefont {Mann}\ and\ \citenamefont {Bremner}(2019)}]{mann2019approximation}%
  \BibitemOpen
  \bibfield  {author} {\bibinfo {author} {\bibfnamefont {R.~L.}\ \bibnamefont {Mann}}\ and\ \bibinfo {author} {\bibfnamefont {M.~J.}\ \bibnamefont {Bremner}},\ }\href {https://doi.org/10.22331/q-2019-07-11-162} {\bibfield  {journal} {\bibinfo  {journal} {Quantum}\ }\textbf {\bibinfo {volume} {3}},\ \bibinfo {pages} {162} (\bibinfo {year} {2019})},\ \Eprint {https://arxiv.org/abs/1806.11282} {arXiv:1806.11282} \BibitemShut {NoStop}%
\bibitem [{\citenamefont {Sinclair}\ \emph {et~al.}(2014)\citenamefont {Sinclair}, \citenamefont {Srivastava},\ and\ \citenamefont {Thurley}}]{sinclair2014approximation}%
  \BibitemOpen
  \bibfield  {author} {\bibinfo {author} {\bibfnamefont {A.}~\bibnamefont {Sinclair}}, \bibinfo {author} {\bibfnamefont {P.}~\bibnamefont {Srivastava}},\ and\ \bibinfo {author} {\bibfnamefont {M.}~\bibnamefont {Thurley}},\ }\href {https://doi.org/10.1007/s10955-014-0947-5} {\bibfield  {journal} {\bibinfo  {journal} {Journal of Statistical Physics}\ }\textbf {\bibinfo {volume} {155}},\ \bibinfo {pages} {666} (\bibinfo {year} {2014})}\BibitemShut {NoStop}%
\bibitem [{\citenamefont {Sly}\ and\ \citenamefont {Sun}(2012)}]{sly2012computational}%
  \BibitemOpen
  \bibfield  {author} {\bibinfo {author} {\bibfnamefont {A.}~\bibnamefont {Sly}}\ and\ \bibinfo {author} {\bibfnamefont {N.}~\bibnamefont {Sun}},\ }in\ \href {https://doi.org/10.1109/focs.2012.56} {\emph {\bibinfo {booktitle} {53rd Annual IEEE Symposium on Foundations of Computer Science (FOCS)}}}\ (\bibinfo {organization} {IEEE},\ \bibinfo {year} {2012})\ pp.\ \bibinfo {pages} {361--369},\ \Eprint {https://arxiv.org/abs/1203.2602} {arXiv:1203.2602} \BibitemShut {NoStop}%
\bibitem [{\citenamefont {Galanis}\ \emph {et~al.}(2016)\citenamefont {Galanis}, \citenamefont {{\v{S}}tefankovi{\v{c}}},\ and\ \citenamefont {Vigoda}}]{galanis2016inapproximability}%
  \BibitemOpen
  \bibfield  {author} {\bibinfo {author} {\bibfnamefont {A.}~\bibnamefont {Galanis}}, \bibinfo {author} {\bibfnamefont {D.}~\bibnamefont {{\v{S}}tefankovi{\v{c}}}},\ and\ \bibinfo {author} {\bibfnamefont {E.}~\bibnamefont {Vigoda}},\ }\href {https://doi.org/10.1017/s0963548315000401} {\bibfield  {journal} {\bibinfo  {journal} {Combinatorics, Probability and Computing}\ }\textbf {\bibinfo {volume} {25}},\ \bibinfo {pages} {500} (\bibinfo {year} {2016})},\ \Eprint {https://arxiv.org/abs/1203.2226} {arXiv:1203.2226} \BibitemShut {NoStop}%
\bibitem [{\citenamefont {Goldberg}\ and\ \citenamefont {Guo}(2017)}]{goldberg2017complexity}%
  \BibitemOpen
  \bibfield  {author} {\bibinfo {author} {\bibfnamefont {L.~A.}\ \bibnamefont {Goldberg}}\ and\ \bibinfo {author} {\bibfnamefont {H.}~\bibnamefont {Guo}},\ }\href {https://doi.org/10.1007/s00037-017-0162-2} {\bibfield  {journal} {\bibinfo  {journal} {Computational Complexity}\ }\textbf {\bibinfo {volume} {26}},\ \bibinfo {pages} {765} (\bibinfo {year} {2017})},\ \Eprint {https://arxiv.org/abs/1409.5627} {arXiv:1409.5627} \BibitemShut {NoStop}%
\bibitem [{\citenamefont {Harrow}\ \emph {et~al.}(2020)\citenamefont {Harrow}, \citenamefont {Mehraban},\ and\ \citenamefont {Soleimanifar}}]{harrow2020classical}%
  \BibitemOpen
  \bibfield  {author} {\bibinfo {author} {\bibfnamefont {A.~W.}\ \bibnamefont {Harrow}}, \bibinfo {author} {\bibfnamefont {S.}~\bibnamefont {Mehraban}},\ and\ \bibinfo {author} {\bibfnamefont {M.}~\bibnamefont {Soleimanifar}},\ }in\ \href {https://doi.org/10.1145/3357713.3384322} {\emph {\bibinfo {booktitle} {Proceedings of the 52nd Annual ACM SIGACT Symposium on Theory of Computing}}}\ (\bibinfo {organization} {ACM},\ \bibinfo {year} {2020})\ pp.\ \bibinfo {pages} {378--386},\ \Eprint {https://arxiv.org/abs/1910.09071} {arXiv:1910.09071} \BibitemShut {NoStop}%
\bibitem [{\citenamefont {Mann}\ and\ \citenamefont {Helmuth}(2021)}]{mann2021efficient}%
  \BibitemOpen
  \bibfield  {author} {\bibinfo {author} {\bibfnamefont {R.~L.}\ \bibnamefont {Mann}}\ and\ \bibinfo {author} {\bibfnamefont {T.}~\bibnamefont {Helmuth}},\ }\href {https://doi.org/10.1063/5.0013689} {\bibfield  {journal} {\bibinfo  {journal} {Journal of Mathematical Physics}\ }\textbf {\bibinfo {volume} {62}},\ \bibinfo {pages} {022201} (\bibinfo {year} {2021})},\ \Eprint {https://arxiv.org/abs/2004.11568} {arXiv:2004.11568} \BibitemShut {NoStop}%
\bibitem [{\citenamefont {Helmuth}\ and\ \citenamefont {Mann}(2023)}]{helmuth2023efficient}%
  \BibitemOpen
  \bibfield  {author} {\bibinfo {author} {\bibfnamefont {T.}~\bibnamefont {Helmuth}}\ and\ \bibinfo {author} {\bibfnamefont {R.~L.}\ \bibnamefont {Mann}},\ }\href {https://doi.org/10.22331/q-2023-10-25-1155} {\bibfield  {journal} {\bibinfo  {journal} {Quantum}\ }\textbf {\bibinfo {volume} {7}},\ \bibinfo {pages} {1155} (\bibinfo {year} {2023})},\ \Eprint {https://arxiv.org/abs/2201.06533} {arXiv:2201.06533} \BibitemShut {NoStop}%
\bibitem [{\citenamefont {Bravyi}\ \emph {et~al.}(2022)\citenamefont {Bravyi}, \citenamefont {Chowdhury}, \citenamefont {Gosset},\ and\ \citenamefont {Wocjan}}]{bravyi2022quantum}%
  \BibitemOpen
  \bibfield  {author} {\bibinfo {author} {\bibfnamefont {S.}~\bibnamefont {Bravyi}}, \bibinfo {author} {\bibfnamefont {A.}~\bibnamefont {Chowdhury}}, \bibinfo {author} {\bibfnamefont {D.}~\bibnamefont {Gosset}},\ and\ \bibinfo {author} {\bibfnamefont {P.}~\bibnamefont {Wocjan}},\ }\href {https://doi.org/10.1038/s41567-022-01742-5} {\bibfield  {journal} {\bibinfo  {journal} {Nature Physics}\ }\textbf {\bibinfo {volume} {18}},\ \bibinfo {pages} {1367} (\bibinfo {year} {2022})},\ \Eprint {https://arxiv.org/abs/2110.15466} {arXiv:2110.15466} \BibitemShut {NoStop}%
\bibitem [{\citenamefont {Apel}\ and\ \citenamefont {Cubitt}(2022)}]{apel2022mathematical}%
  \BibitemOpen
  \bibfield  {author} {\bibinfo {author} {\bibfnamefont {H.}~\bibnamefont {Apel}}\ and\ \bibinfo {author} {\bibfnamefont {T.}~\bibnamefont {Cubitt}},\ }\href@noop {} {\bibfield  {journal} {\bibinfo  {journal} {arXiv e-prints}\ } (\bibinfo {year} {2022})},\ \Eprint {https://arxiv.org/abs/2208.11941} {arXiv:2208.11941} \BibitemShut {NoStop}%
\bibitem [{\citenamefont {Chen}\ \emph {et~al.}(2019)\citenamefont {Chen}, \citenamefont {Galanis}, \citenamefont {Goldberg}, \citenamefont {Perkins}, \citenamefont {Stewart},\ and\ \citenamefont {Vigoda}}]{chen2019fast}%
  \BibitemOpen
  \bibfield  {author} {\bibinfo {author} {\bibfnamefont {Z.}~\bibnamefont {Chen}}, \bibinfo {author} {\bibfnamefont {A.}~\bibnamefont {Galanis}}, \bibinfo {author} {\bibfnamefont {L.~A.}\ \bibnamefont {Goldberg}}, \bibinfo {author} {\bibfnamefont {W.}~\bibnamefont {Perkins}}, \bibinfo {author} {\bibfnamefont {J.}~\bibnamefont {Stewart}},\ and\ \bibinfo {author} {\bibfnamefont {E.}~\bibnamefont {Vigoda}},\ }in\ \href {https://doi.org/10.4230/LIPIcs.APPROX-RANDOM.2019.41} {\emph {\bibinfo {booktitle} {Approximation, Randomization, and Combinatorial Optimization. Algorithms and Techniques (APPROX/RANDOM 2019)}}}\ (\bibinfo {organization} {Schloss Dagstuhl-Leibniz-Zentrum fuer Informatik},\ \bibinfo {year} {2019})\ \Eprint {https://arxiv.org/abs/1901.06653} {arXiv:1901.06653} \BibitemShut {NoStop}%
\end{thebibliography}%

\end{document}